\documentclass[runningheads]{llncs}
\usepackage{amssymb,amsmath}
\usepackage{graphicx,color}
\usepackage{enumerate,paralist}
\usepackage[font=small]{subfig,caption}

\usepackage{todonotes} 
\usepackage{verbatim} 
\usepackage{hyperref}

\let\doendproof\endproof
\renewcommand\endproof{~\hfill\qed\doendproof}

\newcommand{\rephrase}[3]{\noindent\textbf{#1 #2}.~\emph{#3}}

\spnewtheorem{observation}{Observation}{\bfseries}{\itshape}
\newcommand{\R}{\mathrm{I\!R}}

\begin{document}

\title{On Vertex- and Empty-Ply Proximity Drawings}

\titlerunning{On Vertex- and Empty-Ply Proximity Drawings}
%
\author{
P.~Angelini\inst{1} \and
S. Chaplick\inst{2} \and
F. De Luca\inst{3} \and
J. Fiala\inst{4} \and
J. Han\v{c}l Jr. \inst{4} \and
N. Heinsohn \inst{1} \and
M.~Kaufmann\inst{1}  \and
S. Kobourov\inst{5} \and
J. Kratochv\'{\i}l\inst{4} \and
P. Valtr\inst{4} 
}

\institute{%
    Wilhelm-Schickhard-Institut f\"ur Informatik, Universit\"at T\"ubingen, Germany\\
\and
    Lehrstuhl f\"ur Informatik~I, Universit\"at W\"urzburg, Germany\\
\and
	Universit\`{a} degli Studi di Perugia, Perugia, Italy\\
\and
    Department of Applied Mathematics, Charles University (KAM), Czech Republic\\
\and 
    Department of Computer Science, University of Arizona, Tucson, USA
}

%
\authorrunning{Angelini {\em et al.}}
\maketitle

\begin{abstract}
We initiate the study of the \emph{vertex-ply} of straight-line drawings, as a relaxation of the recently introduced \emph{ply} number. 
 Consider the disks centered at each vertex with radius equal to half the length of the longest edge incident to the vertex. The vertex-ply of a drawing is determined by the vertex covered by the maximum number of disks.
The main motivation for considering this relaxation is to relate the concept of ply to proximity drawings. In fact, if we interpret the set of disks as proximity regions, a drawing with vertex-ply number $1$ can be seen as a weak proximity drawing, which we call \emph{empty-ply} drawing.
We show non-trivial relationships between the ply number and the vertex-ply number. Then, we focus on empty-ply drawings, proving some properties and studying what classes of graphs admit such drawings. Finally, we prove a lower bound on the ply and the vertex-ply of planar drawings.
\end{abstract}

\section{Introduction}\label{se:introduction}

Constructing graph layouts that are readable and easily convey the information hidden in the represented data is one of the main goals of graph drawing research.
Several aesthetic criteria have been defined to capture the user requirement for a better understanding of the data, e.g., resolution rules~\cite{fhhklsww-dgphr-90,mp-arpg-92}, low-density~\cite{fr-gdfdp-91}, proximity drawings~\cite{liotta-proximity}. The \emph{ply number}~\cite{ply-original} of a graph is another such criterion. 
We adopt the following notation: given a straight-line drawing $\Gamma$ of a graph $G=(V,E)$, for each vertex $v \in V$ consider an open disk $D_v$ (called the \emph{ply-disk} of $v$) centered at $v$ with radius $r_v$ equal to half of the length of the longest edge incident to $v$. Over all points $p$ on the plane, let $k$ be the maximum number of ply-disks of $\Gamma$  that include the point $p$ in their interior. Then, the drawing $\Gamma$ has \emph{ply} $k$. The \emph{ply number of $G$} is the minimum ply over all its drawings.

The ply number was originally proposed by Eppstein and Goodrich~\cite{eg-snprn-08} in the context of interpreting road networks as
subgraphs of disk-intersection graphs.
The concept of a ply number is also related to proximity drawings of graphs~\cite{liotta-proximity}. 
A \emph{proximity drawing} of a graph $G$ is a straight-line drawing of $G$ in which for every two vertices $u$ and $v$, there exists a region of the plane, called \emph{proximity region} of $u$ and $v$, that contains other vertices in its interior if and only if $u$ and $v$ are not connected by an edge in $G$. If $G$ admits a proximity drawing, then it is a \emph{proximity graph}.
A proximity region specifies a set of points in the plane that are closer to $u$ and $v$ than to the other vertices, 
and different proximity regions lead to different definitions of proximity drawings. Regions can be \emph{global}, e.g., Euclidean minimum spanning trees~\cite{preparata1988computational}, or \emph{local},  e.g., Gabriel graphs~\cite{Gabriel01091969} (Fig.~\ref{fig:proximity-gabriel}), relative-neighborhood graphs~\cite{TOUSSAINT1980261} (Fig.~\ref{fig:proximity-rng}), and Delaunay triangulations~\cite{Del34,preparata1988computational}.
Proximity drawings of graphs are also studied in the \emph{weak} model~\cite{DBLP:journals/jda/BattistaLW06}, where the ``if" part of the condition is neglected: i.e.,  if two vertices are not connected by an edge, then their proximity region may be empty.

In this work, we are interested in deepening the study of the relationship between the notions of ply number and of proximity drawings. In this direction, one can consider the local proximity region associated with a pair of vertices $u$ and $v$ as the one composed of the disks centered at $u$ and at $v$, with radius equal to half of the length of the straight-line segment between $u$ and $v$ (Fig.~\ref{fig:proximity}). Due to the possible absence of edges,
this is a weak proximity model. 
However, a drawing $\Gamma$ may have ply larger than $1$ even if no proximity region contains a vertex different from the two which defined it, since the ply of $\Gamma$ is only determined by the way in which different regions intersect each other.

\begin{figure}[tb!]
	\centering
	\subfloat[\label{fig:proximity-gabriel}]{
		\includegraphics[width=0.18\textwidth,page=5]{images/proximity.pdf}
	}\hfil
	\subfloat[\label{fig:proximity-rng}]{
		\includegraphics[width=0.18\textwidth,page=4]{images/proximity.pdf}
	}\hfil
	\subfloat[\label{fig:proximity}]{
		\includegraphics[width=0.18\textwidth,page=1]{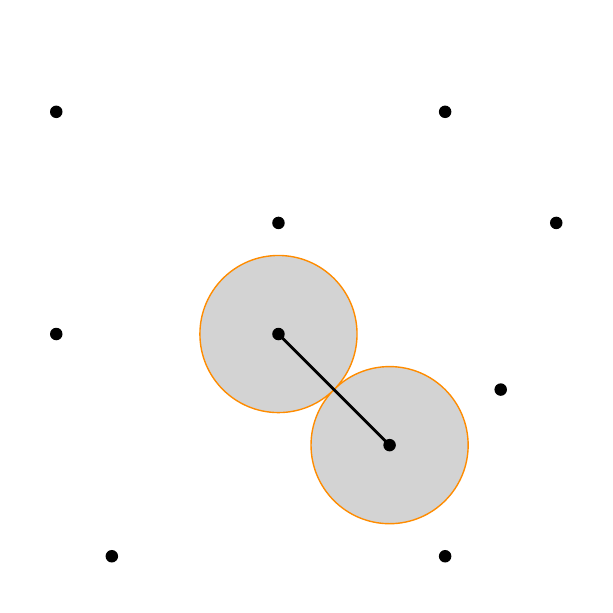}
	}\hfil
	\subfloat[\label{fig:proximity-relationship-1}]{
		\includegraphics[width=0.18\textwidth,page=2]{images/proximity.pdf}
	}\hfil
	\subfloat[\label{fig:proximity-relationship-2}]{
		\includegraphics[width=0.18\textwidth,page=3]{images/proximity.pdf}
	}
	\caption{(a) Gabriel,
	(b)~Relative-neighborhood, and
	(c)~Ply proximity regions.
	(d)~A disconnected empty-ply graph.
	(e)~A non-planar empty-ply drawing.}
	\label{fig:proximity-figures}
\end{figure}
 
To improve this relationship, we relax the definition of ply number and introduce the concept of \emph{vertex-ply number}. Consider a straight-line drawing $\Gamma$ of a graph $G$. Over all vertex-points $p$ on the plane (i.e., points  which realize a vertex of $G$), let $k$ be the maximum number of ply-disks of $\Gamma$  that include the point $p$ in their interior. Then, the drawing $\Gamma$ has \emph{vertex-ply} $k$. The \emph{vertex-ply number of $G$} is the minimum vertex-ply over all its drawings. In the special case in which $\Gamma$ has vertex-ply $1$, i.e., every disk $D_v$ contains only $v$ in its interior, we say that $\Gamma$ is an \emph{empty-ply} drawing. Note that an empty-ply drawing is in fact a weak proximity drawing with respect to the proximity region defined above, that is, a drawing is empty-ply if and only if all the proximity regions are empty.

Some relationships between proximity models are known, e.g., any Delaunay triangulation contains a Gabriel graph as a spanning subgraph, which in turn contains a relative-neighborhood graph, which in turn contains a minimum spanning tree~\cite{liotta-proximity}. 
It is hence natural to ask about the role of empty-ply drawings in these relationships. 
We first note that an empty-ply drawing may be non-planar (see Fig.~\ref{fig:proximity-relationship-2}), which is not the case for Delaunay triangulations, and thus for any of the other type of proximity drawings. On the other hand, there exist empty-ply drawings that are not connected and that cannot be made connected by just adding edges while maintaining the empty-ply property (see Fig.~\ref{fig:proximity-relationship-1}), which differs from
the case for minimum spanning trees (and thus for all the other proximity drawings). These two observations imply that empty-ply drawings are not directly comparable with other types of disk-based proximity drawings. 

The concept of empty-ply is related to \emph{partial edge drawings} (PEDs)~\cite{DBLP:journals/corr/abs-1209-0830,Bruckdorfer2012,DBLP:conf/iisa/Bruckdorfer0L15}. A PED is a straight-line drawing of a graph in which each edge is divided into three segments: a \emph{middle part} that is not drawn and the two segments incident to the vertices, called \emph{stubs}, that remain in the drawing and that are not allowed to cross.  
Our Theorem~\ref{thm:peds} in Section~\ref{sec:properties} shows that an empty-ply drawing also yields a PED whose stubs have nontrivial lengths.

Drawing graphs with low ply was first considered by Di Giacomo {\em et al.}~\cite{ply-original}. They show that testing whether an internally triangulated biconnected planar graph has ply number $1$ can be done in $O(n \log n)$ time and that the class of graphs with ply number $1$ coincides with
unit-disk contact graphs~\cite{Breu19983}, which makes the recognition problem NP-hard. 
Angelini {\em et al.}~\cite{DBLP:conf/gd/AngeliniBBH0KSV16} studied area requirements of drawings of trees with low ply. De Luca {\em et al.}~\cite{DBLP:conf/walcom/LucaGDKL17} performed an experimental study demonstrating correlations between the ply of a drawing and aesthetic metrics such as stress and uniform edge-lengths. An interactive tool has been implemented by Heinsohn and Kaufmann \cite{plyTool}.

We first  demonstrate non-trivial relationships between the ply number and the vertex-ply number of graphs. In Section~\ref{sec:relationship}
we positively answer a question from~\cite{ply-original} (Problem $4$) regarding whether the ply number of an empty-ply drawing is constant.
In Section~\ref{sec:properties} we study properties of empty-ply graphs. In Section~\ref{sec:negative-results} we provide several classes of graphs that admit empty-ply drawings and some classes that do not (we consider $k$-ary trees, complete (bipartite) graphs, and squares of graphs with ply number $1$). Further, in Section~\ref{sec:lower-bound} we answer another question posed in~\cite{ply-original} (Problem $3$), regarding the relationship between (vertex-) ply and crossings, by presenting graphs that admit drawings with constant ply and only $3$ crossings but any corresponding planar drawing requires linear ply. We conclude in Section~\ref{sec:conclusions} with several open problems. 
 
\section{Relationships between Ply and Vertex-Ply}
\label{sec:relationship}

We start with a natural question about the relationship between the ply number and the vertex-ply number of a graph.

\begin{theorem}\label{th:vertex-ply-5k}
	The ply of a drawing of a graph with vertex-ply $h$ is at most $5h$.
\end{theorem}

\begin{proof}
	Let $\Gamma$ be any drawing of a graph $G$ with vertex-ply $h$. 
	Let $p$ be any point in the plane and let $v_1, \ldots, v_k$ be the vertices whose ply-disks contain $p$ in their interior, appearing in this radial order around $p$; see Fig.~\ref{fig:vertex-ply-5k}.
	 Without loss of generality, assume that $v_1$ is the vertex closest to $p$. Let $l$ be the line through $p$ and $v_1$, and let $l'$ and $l''$ be two lines through $p$ creating angles $\frac{\pi}{3}$ and $-\frac{\pi}{3}$ with $l$. These lines determine a covering of the plane by six closed wedges $A_1, \dots, A_6$ centered at $p$, each having $\frac{\pi}{3}$ as its internal angle.
	
	Let $A_1$ and $A_2$ be the wedges delimited by the half-line starting at $p$ and passing through $v_1$. For each vertex $v_i \in A_1\cup A_2$ we have
 $\angle{v_1pv_i} \leq \frac{\pi}{3}$. This implies that $|v_1 v_i| \leq |v_i p|$ and hence that $v_1$ belongs to the ply-disk $D_{v_i}$, since $p$ belongs to $D_{v_i}$. Thus, if the union of the closed wedges $A_1$ and $A_2$ contains at least $h$ vertices among $v_2,\dots,v_k$, we obtain that $v_1$ belongs to at least $h+1$ ply-disks. This is not possible, since $\Gamma$ has vertex-ply $h$.
	
	We now prove that each wedge $A_i$ with $3 \leq i \leq 6$ contains at most $h$ vertices among $v_2, \ldots, v_k$.
	Namely if it contains at least $h+1$ vertices we can argue as above that the closest vertex to $p$ among them belongs to the ply-disks of all the other $h$ vertices. This completes the proof of the theorem that there exist at most $5h$ vertices whose ply-disks enclose $p$.
\end{proof}

	\begin{figure}[tb]
	\centering
	\subfloat[\label{fig:vertex-ply-5k}]{
		\includegraphics[height=3cm,page=2]{images/empty-ply-5.pdf}
	}
	\hfil
	\subfloat[\label{fig:S24}]{
		\includegraphics[height=3cm]{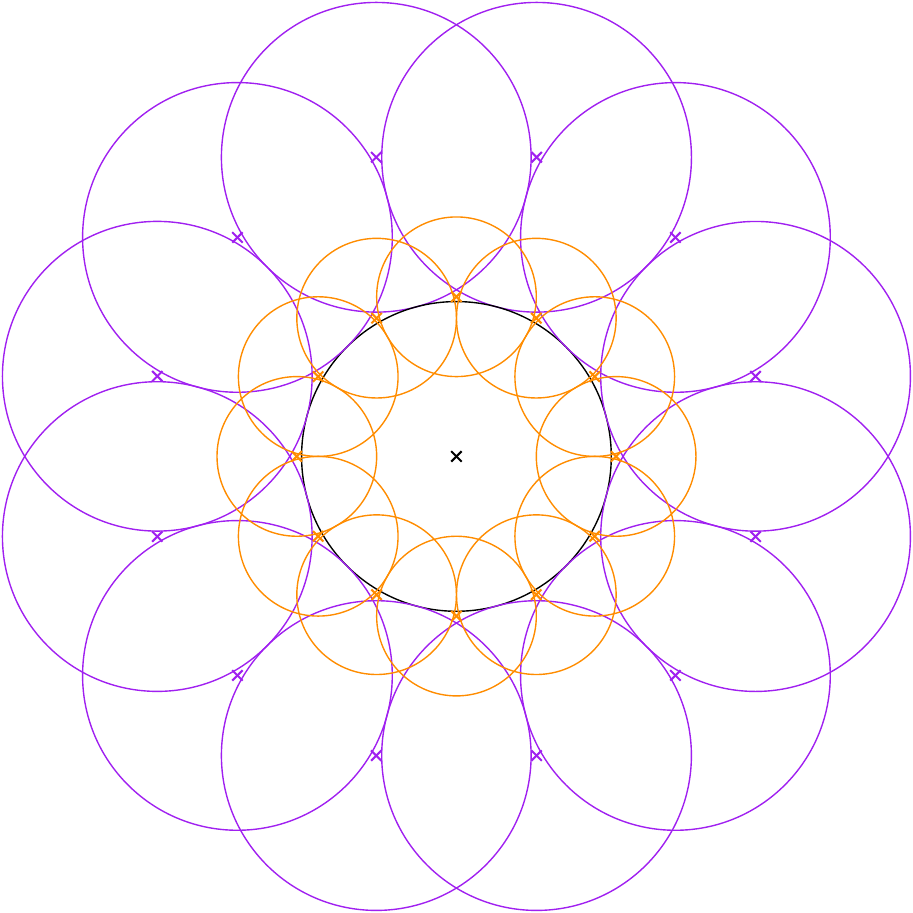}
	}
	\caption{(a) Illustration for the proof of Theorem~\ref{th:vertex-ply-5k}. (b) An empty-ply drawing of a star of degree $24$. For readability, edges are not drawn.}\label{fig:vertex-ply-original}
\end{figure}

\begin{corollary}\label{cor:emply-ply-ordinary-ply}
The ply of an empty-ply drawing of a graph is at most $5$.
\end{corollary}

Note that the converse of Corollary~\ref{cor:emply-ply-ordinary-ply} does not hold. If a graph $G$ does not admit any empty-ply drawing, that does not imply that the ply number of $G$ is larger than $5$. A star graph with degree larger than $24$ does not have an empty-ply drawing (see Theorem~\ref{th:empty-ply-max-degree}), but can be drawn with constant ply $2$~\cite{ply-original}.

\section{Properties of Graphs with Empty-Ply Drawings}
\label{sec:properties}

Let $\Gamma$ be a straight-line drawing of a graph $G$. Let $\{D'_v, v\in V\}$ be the set of open disks where $D_v'$ is centered at $v$, but with radius only $\frac{r_v}{2}$. We can think of these disks as obtained by shrinking the original ply-disks of $\Gamma$ to half-length radius.
Note that if $\Gamma$ is an empty-ply drawing, then the disks in $\{D'_v, v\in V\}$ are pairwise disjoint. This observation implies the next result. 

\begin{lemma}\label{lem:area}
In an empty-ply drawing $\Gamma$ of a graph $G=(V,E)$ the sum of the areas of all ply-disks $\{D_v, v\in V\}$ does not exceed $4$ times the area of their union.  
\end{lemma}

\begin{proof}
Each disk $D_v'$ has area four times smaller than $D_v$, but is drawn inside the union of all ply-disks.
\end{proof}

In the rest of the paper we frequently use disk-packing arguments based on Lemma~\ref{lem:area}. Another consequence of the observation above is a relationship between empty-ply drawings and the most popular type of PED, called $\frac{1}{4}$-SHPED~\cite{Bruckdorfer2012}, in which the length of both stubs of an edge $e$ is $\frac{1}{4}$ of~$e$'s length.

\begin{theorem}\label{thm:peds}
An empty-ply graph admits a $\frac{1}{4}$-SHPED.
\end{theorem}
\begin{proof}
Let $\Gamma$ be an empty-ply drawing of a graph $G=(V,E)$ with the set of disks  $\{D'_v, v\in V\}$. Let $\Gamma'$ be the drawing obtaining from $\Gamma$ by keeping for each edge $(u,v)$ only the two parts in the interior of disks $D_u'$ and $D_v'$. By definition, both these parts cover at least $\frac{1}{4}$ of $(u,v)$. Since no two such disks overlap, there is no crossing in $\Gamma'$, and the statement follows. 
\end{proof}

We now focus on the relationship between the radii of the ply-disks of adjacent vertices in an empty-ply drawing. For the following two lemmas we use that for each vertex $v$, and for each edge $(v,w)$ incident to $v$, we have $r_v \leq |vw|$, as the drawing is empty-ply, and $r_v\ge \frac{|vw|}{2}$, by the definition of the ply-disk $D_v$.

\begin{lemma}\label{le:ratio-incident-edges}
	In an empty-ply drawing, for any two edges $(u,v)$ and $(v,w)$ incident to the same vertex $v$, we have $\frac{1}{2} \leq \frac{|uv|}{|vw|} \leq 2$.
\end{lemma}

\begin{lemma}\label{le:ratio}
In an empty-ply drawing, the radii of the ply-disks of two adjacent vertices $u$ and $v$ differ by at most a factor of $2$, i.e., $\frac{1}{2} \leq \frac{r_u}{r_v} \leq 2$.
\end{lemma}

We conclude the section by presenting a tight bound on the maximum degree of graphs that admit empty-ply drawings.

\begin{theorem}\label{th:empty-ply-max-degree}
No vertex of an empty-ply graph has degree greater than $24$.
\end{theorem}

\begin{proof} 
To obtain a contradiction, let $\Gamma$ be an empty-ply drawing of a graph $G$ with a vertex $v$ of degree greater than $24$. By Lemma~\ref{le:ratio-incident-edges}, the lengths of all edges of $v$ are in the interval $[m, 2m]$, where $m$ is the length of the shortest edge. Note that there are at least $13$ edge lengths either in the interval $[m, \sqrt{2}m]$ or in the interval $[\sqrt{2}m, 2m]$. In either case, there exist two neighbors $u$ and $w$ of $v$ such that $|v u| \le |v w|\le \sqrt{2} |v u| $ and $\alpha=\angle uvw \leq \frac{2\pi}{13}$.
Scaling $\Gamma$ by a factor of $|v u|^{-1}$, we may assume w.l.o.g.\ that $|v u|=1$ and that 
$|v w|=q\in[1,\sqrt{2}]$. By the law of cosines, $|u w|^2= 1+ q^2 -2q \cos\alpha$. As $\Gamma$ is an empty-ply drawing, the vertex $v$ does not belong to the open disk centered at $w$. Hence $|u w| \ge \frac{q}2$.

From the above reasoning it follows that $q$ should satisfy the quadratic inequality $ \frac{q^2}4 \le 1+ q^2 -2q \cos\alpha$, which yields that either $q \leq \frac{4\cos\alpha - \sqrt{16\cos^2\alpha-12}}{3}$ or $q \geq \frac{4\cos\alpha + \sqrt{16\cos^2\alpha-12}}{3}$.
This contradicts the fact that $q\in[1,\sqrt{2}]$, because:
$4\cos\frac{2\pi}{13} - \sqrt{16\cos^2\frac{2\pi}{13}-12}\doteq 2.8 < 3$
and
$4\cos\frac{2\pi}{13} + \sqrt{16\cos^2\frac{2\pi}{13}-12}\doteq 4.27 > 4.24 \doteq 3\sqrt{2}$. This concludes the proof of the theorem.
\end{proof}

Note that $K_{1, 24}$ admits an empty-ply drawing with only two different lengths of edges (see Fig.~\ref{fig:S24}) and so the degree bound provided in Theorem~\ref{th:empty-ply-max-degree} is tight.


\section{Graph Classes with and without Empty-Ply Drawings}
\label{sec:negative-results}


\subsection{Complete Graphs}


\begin{theorem}\label{thm:k8}
Graph $K_n$ admits an empty-ply drawing if and only if $n \leq 7$.
\end{theorem}
\begin{proof} (sketch) 
For a contradiction, suppose that $K_8$ has an empty-ply drawing $\Gamma$. Let $(x_1,x_2)$ be the longest edge of $\Gamma$, w.l.o.g. having length $2$; assume that $x_1$ and $x_2$ lie on an horizontal line $l$. 
Since $(x_1,x_2)$ is the longest edge, the remaining six vertices lie in the intersection of two disks centered at $x_1$ and $x_2$, respectively, with radius $2$; also, by Lemma~\ref{le:ratio-incident-edges}, they lie outside the two disks centered at $x_1$ and $x_2$ with radius $1$; see Fig.~\ref{fig:BasicSplitABCD}. This defines two closed regions in which these vertices lie: one above $l$ and one below.

\begin{figure}[tb]
	\centering
	\subfloat[\label{fig:k7}]{
		\includegraphics[height=3cm,page=1]{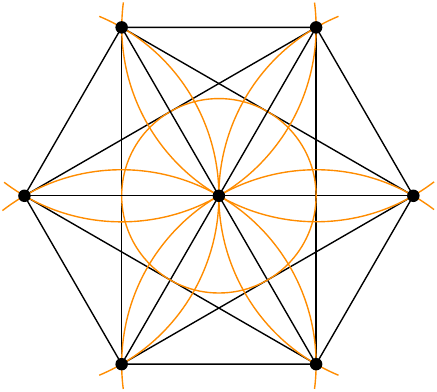}
	}\hfil
	\subfloat[\label{fig:BasicSplitABCD}]{
		\includegraphics[height=3.2cm]{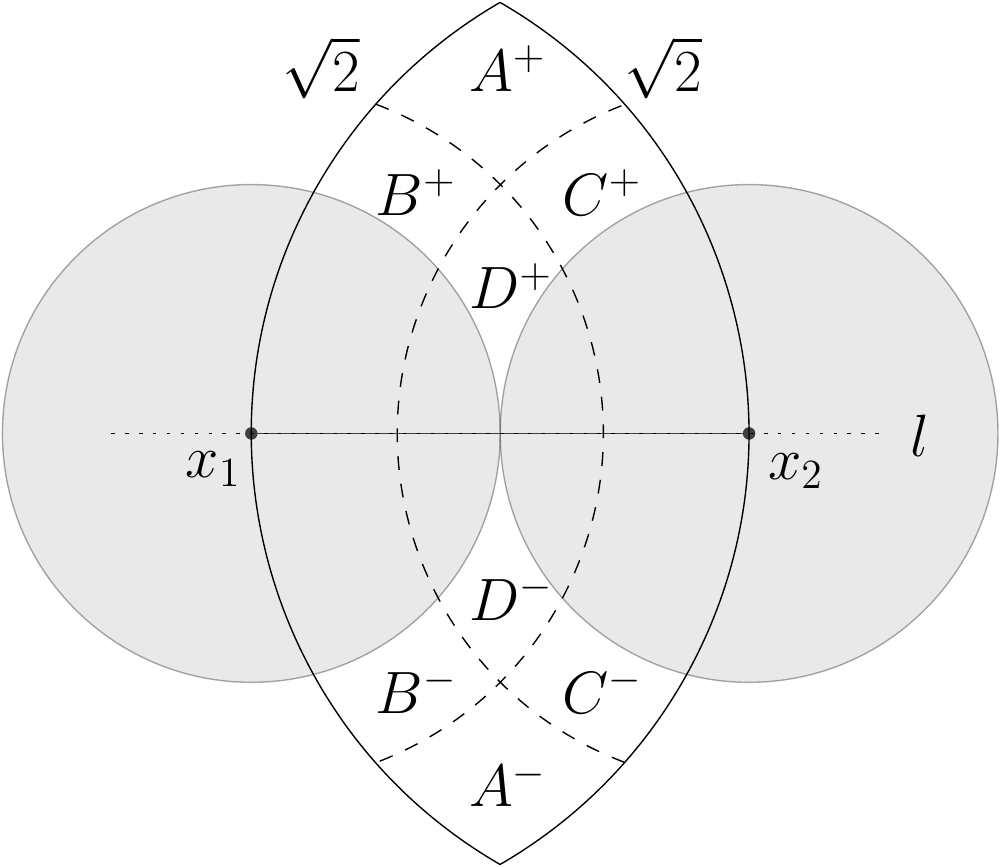}
	}
	\caption{(a) Empty-ply drawing $K_7$; note that there are edges drawn on top of each other. (b) Partition of the region where the vertices of $K_8$ can be placed.
	}\label{fig:complete-graphs-first}
\end{figure}

Using two circles centered in $x_1$ and $x_2$ with radius $\sqrt{2}$, we partition each of these two regions into four closed subregions, called $A^+,B^+,C^+,D^+$ and $A^-,B^-,C^-,D^-$, where the apex $^+$ or $^-$ indicates the region above or below $l$, respectively. Namely, any point in the interior of $A^+ \cup A^-$ (of $D^+ \cup D^-$) has distance larger (smaller) than $\sqrt{2}$ from both $x_1$ and $x_2$; while any point in the interior of $B^+ \cup B^-$ (of $C^+ \cup C^-$) has distance smaller (larger) than $\sqrt{2}$ from $x_1$ and distance larger (smaller) than $\sqrt{2}$ from $x_2$.

We show that any placement of the six remaining vertices in these regions leads to a contradiction.
 We denote by $|X^y|$, with $X\in\{A,B,C,D\}$ and $y\in\{+, -\}$, the number of vertices in $X^y$.
First note that each region can contain at most one vertex, except for $D^+$ and $D^-$, which may contain two vertices.
In fact, if we place any vertex $w$ in a region $X^y$, with $X \in \{A,B,C\}$ and $y \in \{+,-\}$, then the ply-disk $D_w$ of $w$ (defined, at least by the distance to $x_1$, $x_2$) covers the entire region $X^y$.
Regions $D^+$ and $D^-$, on the other hand, have area with height $1$ and width $0.5$. Let $w \in D^+$ be the point at distance $\sqrt{2}$ from both $x_1$ and $x_2$ and $D_w$ be its ply-disk. Then, 
set $D^+ \setminus D_w$ defines an area with diameter at most $\frac13$
and it is not sufficient to place more than one vertex, since the ply disks would have at least a radius $0.5$.

Combining the placement of the vertices in different regions, we can use similar arguments to prove that $|D^+ \cup D^-| \leq 3$ 
 and $|A^+ \cup A^-| \leq 1$.
Also, if $|A^+|=1$ (resp. $|A^-|=1$), then $|D^-|\leq1$ (resp. $|D^+|\leq1$).
 Thus, if $|A^+\cup A^-| = 1$ then  $|D^+\cup D^-| \leq 2$.
Also, if $|A^+|=1$ (resp. $|A^-|=1$) and $|B^-|=1$ (resp. $|B^+|=1$) then either $|B^+|=0$ or $|C^+|=0$ (resp. $|B^-|=0$ or $|C^-|=0$), i.e., $|B^+ \cup C^+| \leq 1$ (resp. $|B^- \cup C^-| \leq 1$).
 By symmetry, if $|A^+|=1$ (resp. $|A^-|=1$) and $|C^-|=1$ (resp. $|C^+|=1$) then either $|B^+|=0$ or $|C^+|=0$ (resp. $|B^-|=0$ or $|C^-|=0$), i.e., $|B^+ \cup C^+| \leq 1$ (resp. $|B^- \cup C^-| \leq 1$).
Hence, if $|A^+ \cup A^-| = 1$, the other regions cannot contain 5 vertices.

The final case where $|A^+ \cup A^-| = 0$ 
directly implies the claim for $K_9$. To prove this for $K_8$ we can see that
 if $|B^-|=1$ and $|C^+|=1$ (resp. $|B^+|=1$ and $|C^-|=1$), then $|D^+\cup D^-| \leq 1$, which again leads to a contradiction.
To conclude the proof, we present an empty-ply drawing for $K_7$ in Fig.~\ref{fig:k7}.
We strongly believe that this drawing is unique.
\end{proof}

\subsection{Complete bipartite graphs}

We now consider complete bipartite graphs. For proof-by-picture of the next theorem see Fig.~\ref{fig:S24} and Figs.~\ref{fig:k-2-12}-\ref{fig:k-4-6}.

\begin{figure}[tb!]
	\centering
	\subfloat[\label{fig:k-2-12}]{
		\includegraphics[height=3cm]{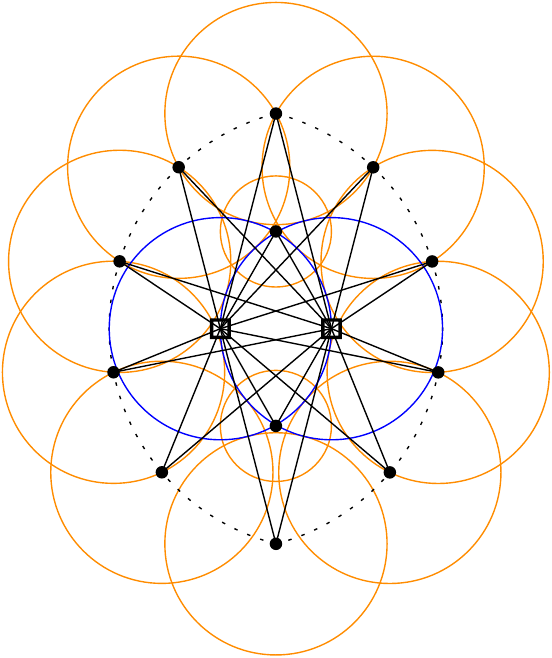}
	}\hfil
	\subfloat[\label{fig:k-3-9}]{
		\includegraphics[height=3cm]{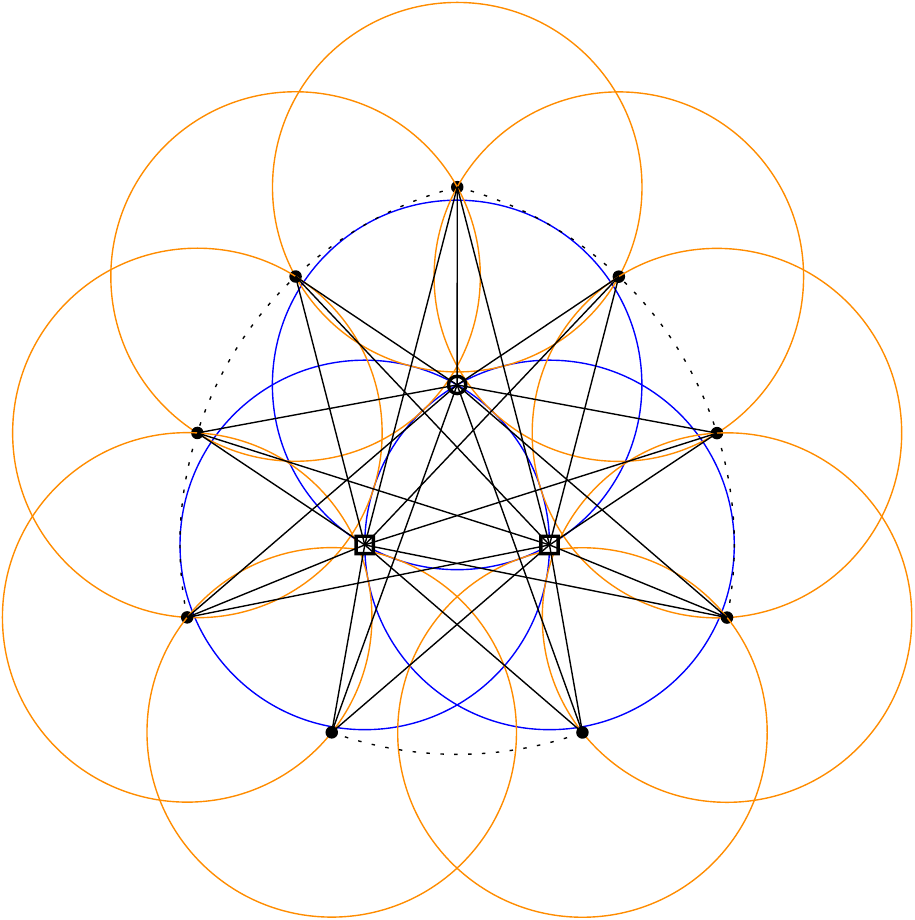}
	}\hfil
	\subfloat[\label{fig:k-4-6}]{
		\includegraphics[height=3cm]{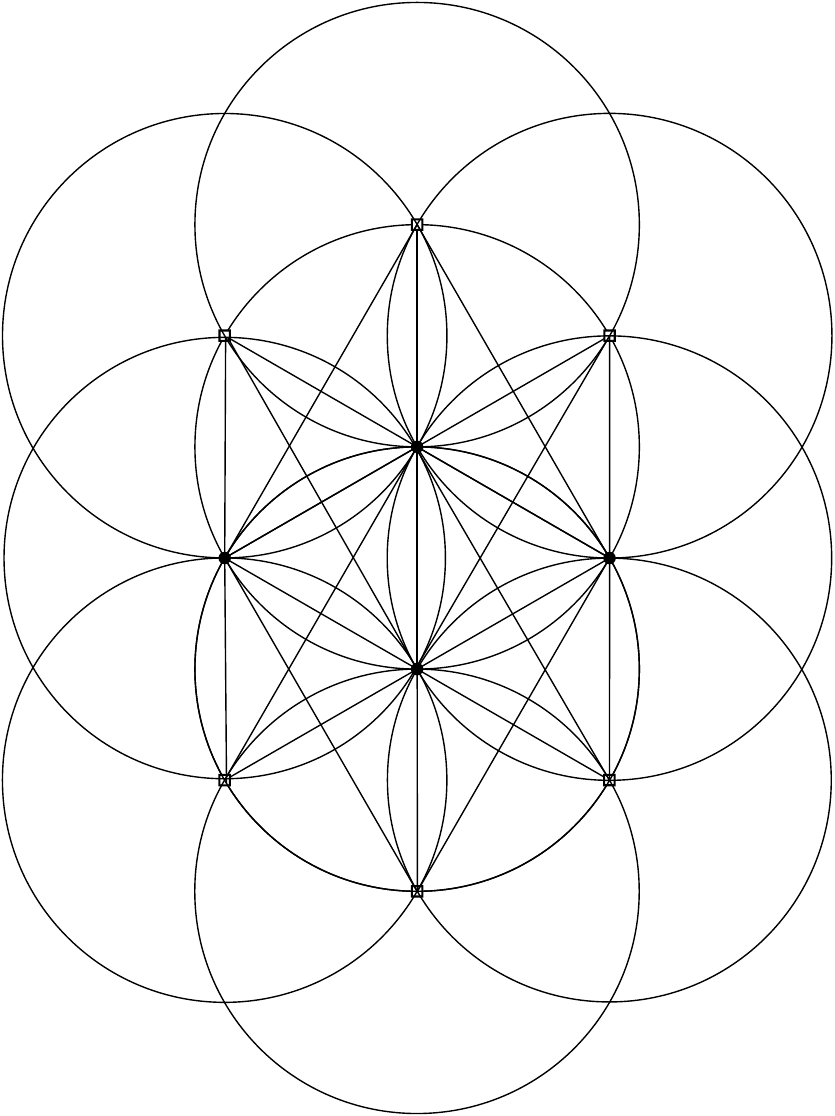}
	}
	\caption{Empty-ply drawing of (a) $K_{2,12}$, (b) $K_{3,9}$, and (c) $K_{4,6}$. Note that the drawing of $K_{4,6}$ has edges drawn on top of each other. 
}
	\label{fig:complete-graphs}
\end{figure}

\begin{theorem}\label{thm:complete-bipartite}
	Graphs $K_{1,24}$, $K_{2,12}$, $K_{3,9}$, and $K_{4,6}$ admit empty-ply drawings.
\end{theorem}

Note that Theorem~\ref{th:empty-ply-max-degree} implies that $K_{1,25}$ does not admit any empty-ply drawing, and hence this is true for any complete bipartite graph $K_{n,m}$ with $n$ or $m$ greater than $24$. This leaves a wide open gap between the upper bounds on the values of $n$ and $m$, and the lower bounds from Theorem~\ref{thm:complete-bipartite}. 

For $K_{2,m}$, we give a negative result for $m \geq 15$ in the following theorem
based on arguments similar to those in Theorem~\ref{thm:k8}. 

\begin{theorem}\label{thm:bipartiteK2X}
	Graph $K_{2,m}$ with $m \geq 15$ does not admit any empty-ply drawing.
\end{theorem}

\subsection{Trees of Bounded Degree}

A \emph{$d$-ary tree $T$ with $k$ levels} is a rooted tree where all vertices at distance less than $k$ from the root have at most $d$ children and the remaining ones are leaves. If all the non-leaf vertices have exactly $d$ children, we say that $T$ is \emph{complete}. Any tree with maximum degree $\Delta$ is a subtree of a $(\Delta-1)$-ary tree.

Note that binary trees admit empty-ply drawings, as the drawings with ply $2$ constructed by the algorithm in~\cite{ply-original} are empty-ply drawings. 
Applying  Corollary~\ref{cor:emply-ply-ordinary-ply} 
to the class of complete $10$-ary trees (which do not admit drawing with constant ply~\cite{DBLP:conf/gd/AngeliniBBH0KSV16}) shows that they do not admit empty-ply drawing. But we can prove something stronger.

\begin{theorem}\label{thm:empty-4-ary}
For sufficiently large $k$, the complete $4$-ary tree $T_k$ with $k$ levels admits no empty-ply drawing.
\end{theorem}

\begin{proof}
Assume without loss of generality that $k$ is even and that $T_k$ has an empty-ply drawing $\Gamma$ where the ply-disk of the root $v_0$ has unit radius. We announce that for simplicity the following estimates are not stated in the tightest form.
We will make use of the following consequences of Lemmas~\ref{le:ratio-incident-edges} and \ref{le:ratio}:

\begin{claim}[A] 
If a ply-disk of a vertex $u$ in $\Gamma$ has radius at least $2^i$, then all the leaves of the subtree rooted at $u$ have radii at least $2^{2i-k}$.
\end{claim} 

\begin{proof}
Since $r_u \geq 2^i$, the distance between $u$ and the root is greater than $i$ by Lemma \ref{le:ratio}. Thus the path from $u$ to its leaves has length at most $k-i$.
\end{proof}

\begin{claim}[B] 
If $v$ is a leaf whose ply-disk has radius $r_v\in (2^{2i-k},2^{2i-k+2}]$, with $i\in\{0,k-1\}$, then its Euclidean distance from the root is $|v_0 v|\leq 2^{i+2}$.
\end{claim} 

\begin{proof}
Let $v_0,v_1,\dots,v_k=v$ be the path from the root $v_0$ to $v_k$ in $T_k$.
Since $r_{v_0}=1$, edge $(v_0,v_1)$ has length at most $2$. Also, by Lemma~\ref{le:ratio-incident-edges}, the lengths of the edges can grow at most by a factor $2$ along the path; hence, $|v_{j-1} v_j|\leq 2^j$ for $j\in\{1,\dots,i\}$. 
If we traverse the path in the opposite direction from $v_k$, whose ply-disk has radius at most $2^{2i-k+2}$, we get analogously that $|v_{k-j+1} v_{k-j}|\leq 2^{j+2i-k+2}$ for $j\in\{1,\dots,k-i\}$. 

The total distance is thus bounded by $|v_0 v_k| \leq |v_0 v_1|+|v_1 v_2|+\dots+|v_{k-1} v_k|=$ $\sum_{j=1}^{i} |v_{j-1} v_j| + \sum_{j=1}^{k-i} |v_{k-j+1} v_{k-j}|
\leq \sum_{j=1}^{i} 2^j + \sum_{j=1}^{k-i}  2^{j+2i-k+2} = 2^{i+1}-2 +2^{i+1}- 2^{3+2i-k} \ \leq \ 2^{i+2}$, and the statement follows.
\end{proof}

We now distribute the $4^k$ leaves to $k$ sets $L_0,\dots,L_{k-1}$ (all logarithms binary):

\begin{itemize}
\item[a)] if $i\geq 3\log k$ then $L_i=\{v: r_v\in (2^{2i-k},2^{2i-k+2}]\}$
\item[b)] if $i< 3\log k$ then $L_i=\{v: r_v \leq 2^{6\log k-k}$  and whose largest predecessor $u$ has radius  $r_u\in [2^i,2^{i+1})\}$
\end{itemize}

In the first case, the radii of the leaves in $L_i$ are sufficient to obtain a good bound on enclosing area of the disks in $L_i$. In the other case, the radius on the enclosing disk for $L_i$ mostly depends on the presence of predecessors that are larger than the root disk.

Some of the sets are empty by the definition, but it is irrelevant to our further deductions. By pigeonhole principle, either some $L_i$, $i\geq 3\log k$ satisfies $|L_i| \ge \frac{4^k}{2 k}$ or some $L_i$, $i<3\log k$ satisfies $|L_i| \ge \frac{4^k}{6\log k}$, since $\frac{k+1 - 3 \log k}{2 k}+\frac{3\log k}{6\log k}\leq 1$ when $k\ge \sqrt[3]{2}$.

The rough idea behind the distinction of these two cases is that in case a), when the diameters of leaves are sufficiently large, it suffices to consider twice smaller proportion than the uniform pigeonhole principle would use and show that the total area of ply-disks corresponding to leaves of $L_i$ is still too large for an empty-ply drawing $\Gamma$. In case b) we use a slightly more elaborate argument considering also the area of the predecessors of the vertices in $L_i$.
\paragraph{Case a)} Assume that for some $i\geq 3\log k$ it holds that $|L_i| \ge \frac{4^k}{2 k}$. The total area occupied by the disks in $L_i$ is at least $\frac{4^k}{2 k} \pi 4^{2i-k}=\frac{8^i\pi}{2 k}$.
By Claim (B), for every $v\in L_i$ it holds that $|v_0 v|\leq 2^{i+1}$, hence all 
ply-disks of $L_i$ must be contained in a disk centered at the root of radius $2^{i+1}+2^{2i-k+2} \leq 5\cdot 2^i$, since for $i\in\{0,\dots,k\}: i>2i-k$. In particular this disk has area at most $25 \pi 4^i$.

In order to apply Lemma~\ref{lem:area}, it suffices to choose $k$ large enough such that $\frac{8^i\pi}{2k} > 4 \cdot 25 \pi 4^i$ for all $i \geq 3\log k$, i.e., $k > \sqrt[5]{200} 
\doteq 2.9$.
\paragraph{Case b)} Assume that for some $i < 3\log k$ it holds $|L_i| \ge \frac{4^k}{6\log k}$. Any $v\in L_i$ has radius smaller than $2^{6\log k-k}$, as otherwise we would be in case a).
To obtain the maximum distance between $v$ and the root $v_0$ we argue that the first $3\log k$ disks along the path from $v_0$ to $v$ may have radius at most $2^{i+1}$. Analogously as in the proof of Claim (B), the $j$-th predecessor of $v$ has radius at most $2^{6\log k-k+j}$. An upper bound of $|v_0 v|\le 2^{i+2} (3\log k + 1)$ is obtained by summing up. 

We now consider the subtree $T'$ of $T_k$ induced by the vertices of $L_i$ and all their predecessors. Note that the drawing of the entire tree $T'$ shall be contained within a disk of radius $2^{i+2} (3\log k + 1) + 2^{3\log k-k}$, i.e., in area at most $4^{i+3}\pi$.
On the other hand, by Claim (A), each of the leaves has radius at least $2^{2i-k}$ . Thus, their total area is at least $\frac{4^k}{6\log k}4^{2i-k}\pi=\frac{8^i\pi}{6\log k}$.

The number of parents of disks in $L_i$ is at least $\frac{4^{k-1}}{6\log k}$, each of radius at least $2^{i-1}$, hence they occupy area also bounded from below by $\frac{8^i\pi}{6\log k}$. Thus, all leaves in $L_i$ and all their $k-i$ predecessors occupy space at least $\frac{8^i\pi}{6\log k}(k-i) \ge \frac{8^i\pi k}{12 \log k}$.
Again, to apply Lemma~\ref{lem:area}, it suffices to choose $k$ large enough such that $\frac{8^i\pi k}{12 \log k} > 4 \cdot 64 \pi 4^i$ for all non-negative $i < 3\log k$ (in particular for $i=0$). A straightforward calculation verifies that the inequality holds e.g., for $k \ge 2^{16}$.  

For $k=2^{16}$ 
one of the two cases applies, which concludes the proof.
\end{proof}

Theorem~\ref{thm:empty-4-ary} leaves open the question for $3$-ary trees. We remark that the algorithm for binary trees~\cite{ply-original} adopts a common drawing style: the orthogonal one with a shrinking factor of $1/2$; see also~\cite{efk-ogd-99}. We prove that this technique fails for $3$-ary trees, for any shrinking factor in $(0,1)$.

\begin{theorem}\label{thm:ternary-shrinking}
Rooted ternary trees do not admit empty-ply drawings constructed in orthogonal fashion with shrink factor $q$ for any $q\in(0,1)$, i.e., when the distance from a vertex to its children is $q$ times the distance to its parent.
\end{theorem}

\subsection{Graph Squares}

The \emph{square} of a graph $G$ is the graph obtained from $G$ by adding an edge between each vertex and the neighbors of its neighbors.

\begin{theorem}\label{le:square}
Let $G^2$ be the square of a graph $G$. If $G$ admits a drawing with ply $1$, then $G^2$ admits an empty-ply drawing. Also, if $G$ is a subgraph of a triangular tiling, then $G^2$ admits an empty-ply drawing with ply at most $4$.
\end{theorem}

\begin{proof}
Let $\Gamma$ be a straight-line drawing of $G$ with ply $1$. As proved in~\cite{ply-original}, all the edges of $G$ have the same length, say $1$, in $\Gamma$, and every two non-adjacent vertices are at distance at least $1$ from each other. Hence, adding the edges of $G^2 \setminus G$ to $\Gamma$ produces a drawing $\Gamma^2$ of $G^2$ in which each edge has length at most $2$. This implies that every ply-disk has radius at most $1$ in $\Gamma^2$, and thus $\Gamma^2$ is an empty-ply drawing. Note that $\Gamma^2$ may contain edge overlaps. 

For the second part of the statement, recall that if $G$ is a subgraph of a triangular tiling, then it admits a drawing $\Gamma$ in which all edges have the same length and all the angles are multiples of $\frac{\pi}{3}$. Hence, $\Gamma$ has ply $1$. Also the drawing $\Gamma^2$ obtained by adding the additional edges of $G^2 \setminus G$ to $\Gamma$ is an empty-ply drawing.  In this case, however, we can also prove that the ply of $\Gamma^2$ is at most $4$; recall that an upper bound of $5$ to the ply of $\Gamma^2$ is already implied by Corollary~\ref{cor:emply-ply-ordinary-ply}.

W.l.o.g.\ let the triangular tiling be of unit edge length. Consider the open disk of unit radius, which is centered at an arbitrary point $p$ on the plane. If $p$ is not a vertex of the triangular tiling, at most four vertices of the triangular tiling may fall in this disk.  
In the case where $p$ is a vertex of the triangular tiling, no other vertex of the tiling falls in the disk, but only on its boundary. Thus, any point $p$ can be internal to at most four ply-disks of the tiling vertices.
\end{proof}

\section{Ply and Vertex-Ply of Planar Drawings}
\label{sec:lower-bound}

In the original paper on the ply number it
 was observed that considering only plane graph drawings may prevent finding low ply non-plane drawings~\cite{ply-original}. 
In particular, for the class of \emph{nested-triangles} graphs 
the ``most natural'' planar drawing has ply $\Omega(n)$ (see Fig~\ref{fig:nested-triangles-natural}), while there exist non-planar drawings (with edge overlaps) with ply $5$ (see Fig.~\ref{fig:nested-triangles-non-planar}).
Note that however a ``less natural'' planar drawing with ply $4$ can always be constructed; see Fig.~\ref{fig:nested-triangles-2}. 
\begin{figure}[t!]
	\centering
	\subfloat[\label{fig:nested-triangles-natural}]{
		\includegraphics[height=3cm,page=5]{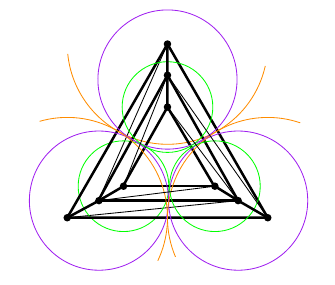}
	}\hfil
	\subfloat[\label{fig:nested-triangles-non-planar}]{
		\includegraphics[height=3cm,page=8]{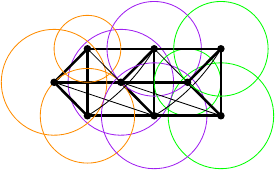}
	}\hfil
	\subfloat[\label{fig:nested-triangles-2}]{
		\includegraphics[height=3cm,page=4]{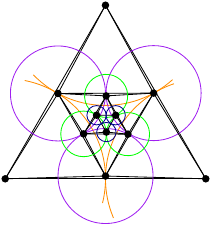}
	}
	\caption{Nested triangles graph:
	(a) ``The most natural'' drawing. 
	(b) A non-planar drawing with ply 5. 
	(c) A planar drawing with ply $4$. The disks of three vertices at the same level do not properly overlap, and disks at levels $i$ and $i+3$ do not overlap.}
	\label{fig:nested}
\end{figure}

We strengthen this observation by providing a planar $3$-tree $G$ admitting a non-planar drawing (with only $3$ crossings) with ply $5$, such that \emph{any} planar drawing of $G$ has ply $\Omega(n)$; the same linear lower bound holds even for vertex-ply when the outer face is fixed. Recall that a \emph{planar $3$-tree} can be constructed, starting from a $3$-cycle, by repeatedly adding a vertex inside a triangular face and connecting it to all three vertices of this face.

Our result also gives a negative answer to an open question posed in~\cite{ply-original} on whether there exists a relationship between the number of crossings and the ply number of a drawing. Our example shows that one can reduce the ply number from $\Omega(n)$ to $O(1)$, by introducing only $O(1)$ crossings.

\begin{theorem}\label{th:linear-lower-bound}
There exists an $n$-vertex planar $3$-tree $G$ such that any planar drawing of $G$ with a fixed outer face has vertex-ply $\Theta(n)$, and hence ply $\Theta(n)$, while $G$ admits a drawing with ply $5$ and vertex-ply $4$ with three edge crossings.
\end{theorem}
\begin{proof}

Graph $G$ has three vertices $v_1$, $v_2$, and $v_3$ on the outer face, and a vertex $u$ that is connected to all of $v_1$, $v_2$, and $v_3$. Refer to Fig.~\ref{fig:theta-n-graph}. 
In addition, it contains three paths $x_1,\dots,x_m$, $y_1,\dots,y_m$, and $z_1,\dots,z_m$, each on $m=\frac{n-4}{3}$ vertices. The edge set further contains edges $(u,x_1)$, $(u,y_1)$, $(u,z_1)$ and also $(x_i,v_1)$, $(x_i,v_2)$, $(y_i,v_2)$, $(y_i,v_3)$, $(z_i,v_1)$, $(z_i,v_3)$ for each $i\in \{1,\dots,m\}$.
	
Consider any planar drawing $\Gamma$ of $G$. Suppose, w.l.o.g., that $(v_1,v_2)$ is of unit length and that it is the longest edge in $\Gamma$ among the three edges incident to the outer face, that is, $|v_2 v_3|, |v_1 v_3| \leq  1$. 
Since vertex $u$ lies inside the triangle $v_1 v_2 v_3$, we have $|u v_1|, |u v_2| < 1$. Hence, it is possible to cover the whole region of the plane delimited by triangle $u v_1 v_2$ with a set of $28$ disks, each having radius $\frac{1}{8}$, as illustrated in Fig.~\ref{fig:disk-cover}. Thus, at least one disk $D$ out of these 28 contains in its interior at least $\frac{m}{28}=\frac{n-4}{84}$ vertices out of $x_1,\dots,x_m$. 

\begin{figure}[t!]
	\centering
	\subfloat[\label{fig:theta-n-graph}]{
		\includegraphics[height=2.5cm,page=1]{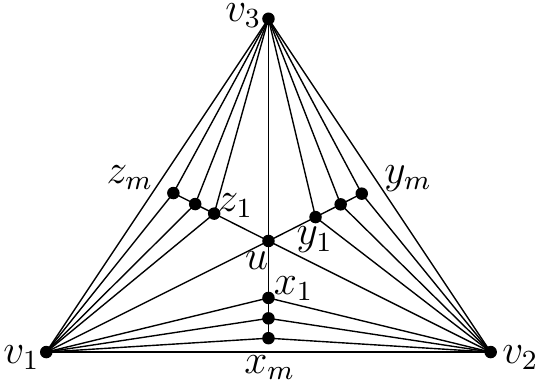}
	}\hfil
	\subfloat[\label{fig:disk-cover}]{
		\includegraphics[height=2.5cm,page=1]{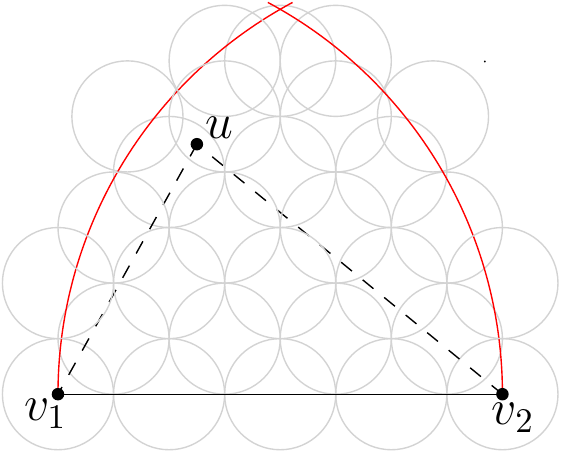}
	}\hfil
	\subfloat[\label{fig:theta-n-graph-non-planar}]{
		\includegraphics[height=3cm,page=2]{images/theta-n-graph.pdf}
	}
	\caption{(a) The planar $3$-tree $G$ in the proof of Theorem~\ref{th:linear-lower-bound}. (b) A set of $28$ disks of radius $\frac{1}{8}$ covering the whole region delimited by triangle $u v_1 v_2$ when $|v_1 v_2| = 1 > |u v_1|,|u v_2|$. (c) A non-planar drawing of $G$ with ply $5$ and vertex-ply $4$.}
	\label{fig:linear-lower-bound}
\end{figure}

Consider any vertex $x_i\in D$. Since $x_i$ is connected to both $v_1$ and $v_2$, the longest of its incident edges has length at least $\frac{1}{2}$, and hence the radius of the ply-disk of $x_i$ is at least $\frac{1}{4}$. 
Hence the ply-disk of $x_i$ entirely contains the disk $D$ in its interior, and thus it contains all the vertices inside it. 
Since this is true for all the $\frac{n-4}{84}$ vertices inside $D$, the first part of the statement follows.
	
A non-planar drawing of $G$  with ply $5$ and vertex-ply $4$ is depicted in Fig.~\ref{fig:theta-n-graph-non-planar}. Here vertices $v_1, v_2$ and $v_3$ form an equilateral triangle with barycenter $u$.
Vertices $x_1,\dots,x_m$ are arranged along the axis of the segment $v_1v_2$ at distances growing exponentially by a factor of $2$, analogously for vertices $y_1,\dots,y_m$ and $z_1,\dots,z_m$. The disk $D_u$ overlaps with $D_{x_1}$, $D_{y_1}$, and $D_{z_1}$, without enclosing these vertices. 
The drawing of the subset of vertices  $\{u,x_1,y_1,z_1\}$
is empty-ply and of ply $2$. After considering the remaining vertices, the disks of $v_1$, $v_2$, $v_3$ may contain all of them in their interior. Thus we obtain ply $5$ and vertex-ply $4$. 
\end{proof}

\section{Conclusions and Future Work}
\label{sec:conclusions}

We defined and studied the vertex-ply of a straight-line drawing, paying  particular attention to the special case of empty-ply drawings, whose vertex-ply is $1$. 
We conclude with several natural open problems.

\begin{enumerate}
\item We know that binary trees admit empty-ply drawings~\cite{ply-original} and that $4$-ary trees do not (Theorem~\ref{thm:empty-4-ary}). What about $3$-ary trees? Note that Theorem~\ref{thm:ternary-shrinking} rules out a large class of possible drawings (orthogonal and shrinking).
\item
Another way of generalizing binary trees is to maintain the degree restriction, leading to the question: do (planar) max-degree-$3$ graphs admit empty-ply drawings? 
\item In Theorem~\ref{le:square} we proved that the square $G^2$ of a graph $G$ with ply $1$ admits an empty-ply drawing, which has ply at most $5$ by Corollary~\ref{cor:emply-ply-ordinary-ply}. 
On the other hand, if $G$ is a subgraph of a triangular tiling, then the empty-ply drawing of $G^2$ has ply at most $4$. Does the square of every graph with ply $1$ admit an (empty-ply) drawing with ply $4$? 
Note that there are ply $1$ graphs that are not subgraphs of a triangular tiling.
\item Looking at empty-ply drawings from the proximity perspective, it is natural to consider the generalization in which ply-disks do not need to be empty, but can contain at most $k$ vertices. We call a drawing with this property a \emph{$k$-empty-ply} drawing, in compliance with the definition of $k$-Gabriel and $k$-relative-neighborhood drawings~\cite{liotta-proximity}. With the argument of Theorem~\ref{th:linear-lower-bound} there exist $n$-vertex graphs whose any planar drawing is $\Omega(n)$-empty-ply.
\item In Theorem~\ref{thm:k8} we proved a tight bound of $7$ on the size of complete graphs admitting empty-ply drawings. For complete bipartite graphs $K_{n,m}$, we have a tight bound of $m=24$, for $n=1$, and an almost tight bound of $12 \leq m \leq 14$, for $n=2$, with larger gaps between the bounds for larger values of $n$. 
\end{enumerate}

\paragraph*{Acknowledgments}
This work began at the 2015 HOMONOLO meeting. We gratefully thank M.~Bekos, T.~Bruckdorfer, G.~Liotta, M.~Saumell, and A.~Symvonis for great discussions on the topic.
Research was partially supported by project CE-ITI P202/12/G061 of GA\v{C}R (J.F., J.K., P.V.), by project SVV--2017--260452 (J.H.), and by DFG grant Ka812/17-1 (P.A., N.H., M.K.).

\bibliographystyle{splncs03}
\bibliography{biblio}

\newpage
\section*{Appendix}

\section*{Complete proof of Theorem~\ref{thm:k8}}
In this section we give a complete proof of Theorem~\ref{thm:k8}, which we restate here for the reader's convenience.

\medskip
\rephrase{Theorem}{\ref{thm:k8}}{
Graph $K_n$ admits an empty-ply drawing if and only if $n \leq 7$.
}
\medskip

For the contrary we assume that there exists an empty ply drawing of $K_n$ where $n \geq 8$.
We state observations regarding the edge $(x_1, x_2)$ and placement of vertices in specified regions and conclude the Theorem by lemmas where we distinguish the different cases about the number of vertices placed below the edge $(x_1, x_2)$. 

Let $\Gamma$ be an \textit{empty-ply} drawing of $K_n$ where $e = (x_1, x_2)$ is the longest edge of $\Gamma$ with length $2$. Thus all other vertices lie in the intersection of the annuli centred at $x_1, x_2$ with radii 1 and 2. 
Without loss of generality we draw this edge as horizontal line segment as in Figure \ref{fig:BasicSplit}.\\
We partition the intersection is as follows: 
\begin{align*}
A &= \{x \in \R^2 : \sqrt{2} < |x_1x|\leq 2 , \sqrt{2} < |x_2x| \leq 2\},\\
B &= \{x \in \R^2 : 1 \leq |x_1x|\leq \sqrt{2}, \sqrt{2} \leq |x_2x| \leq 2\},\\
C &= \{x \in \R^2 : \sqrt{2} \leq |x_1x|\leq 2, 1 \leq |x_2x| \leq \sqrt{2}\},\\
D &= \{x \in \R^2 : 1 \leq |x_1x|\leq \sqrt{2}, 1 \leq |x_2x| \leq \sqrt{2}\}.\\
\end{align*}

\begin{figure}
\begin{center}
\includegraphics[width=.5\textwidth]{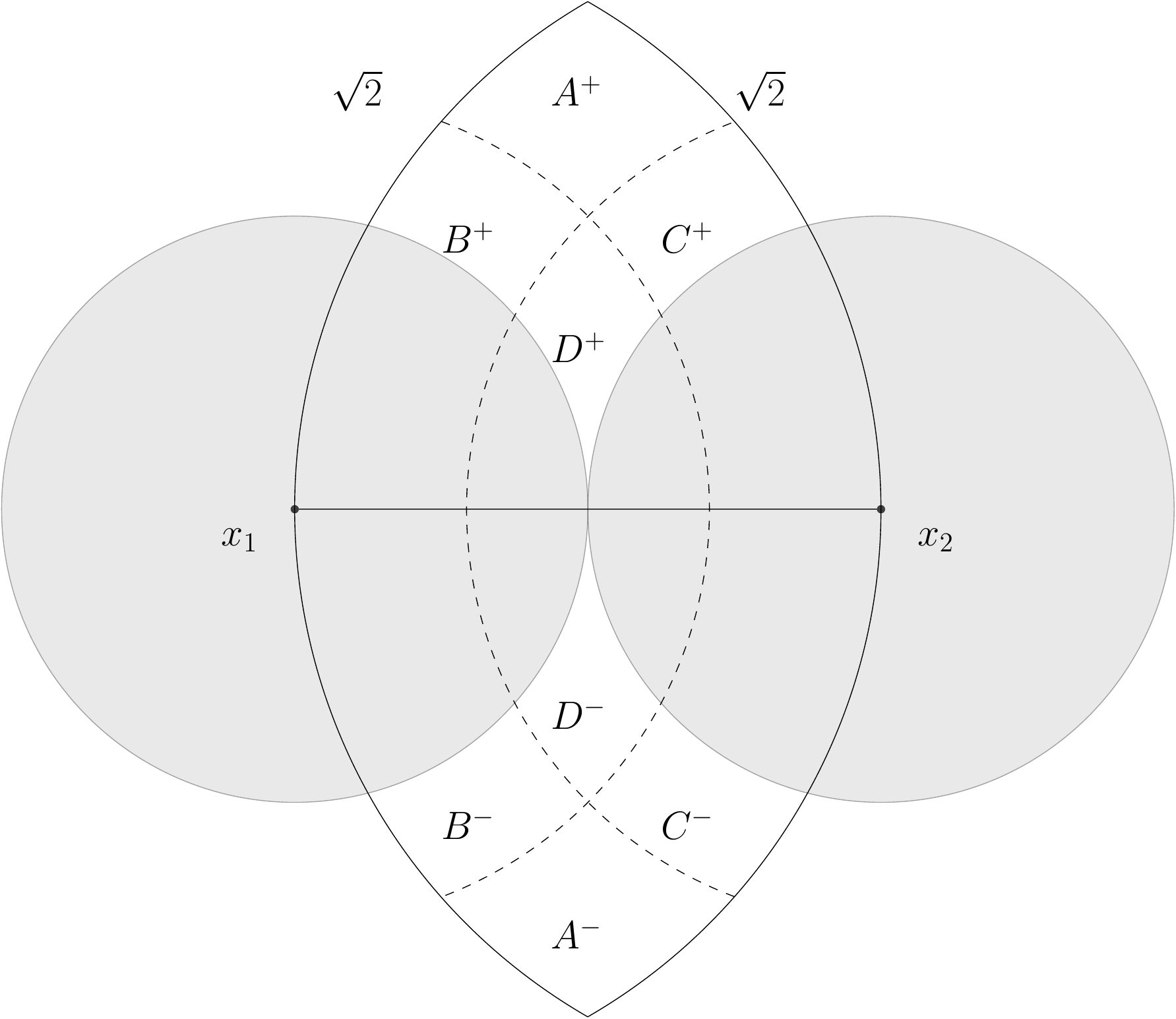}
\caption{Where $+$ is above the edge $(x_1, x_2)$ and $-$ below respectively.}\label{fig:BasicSplit}
\end{center}
\end{figure}

\begin{observation}\label{Obs:B+=1}
There is at most one vertex in $B^+$ and by symmetry in $B^-$ as well as in $C^+$ and $C^-$.
\end{observation}
\begin{proof}
Note that the diameter of $B^+$ meaning the maximal distance of any two points  in $ B^+$ is $\leq 0.75$ (see Figure \ref{fig:diameterB}).\\
We split the area of $B^+$ into two parts like indicated in Figure \ref{fig:splitB1B2}:
\begin{align*}
B_1 &= \{x \in B : \sqrt{2} \leq |x_2x| \leq \sqrt{3} \},\\
B_2 &= \{x \in B : \sqrt{3} \leq |x_2x| \leq 2 \}.
\end{align*}

\begin{figure}
 \begin{minipage}[t][][t]{.5\textwidth}
\includegraphics[width=\textwidth]{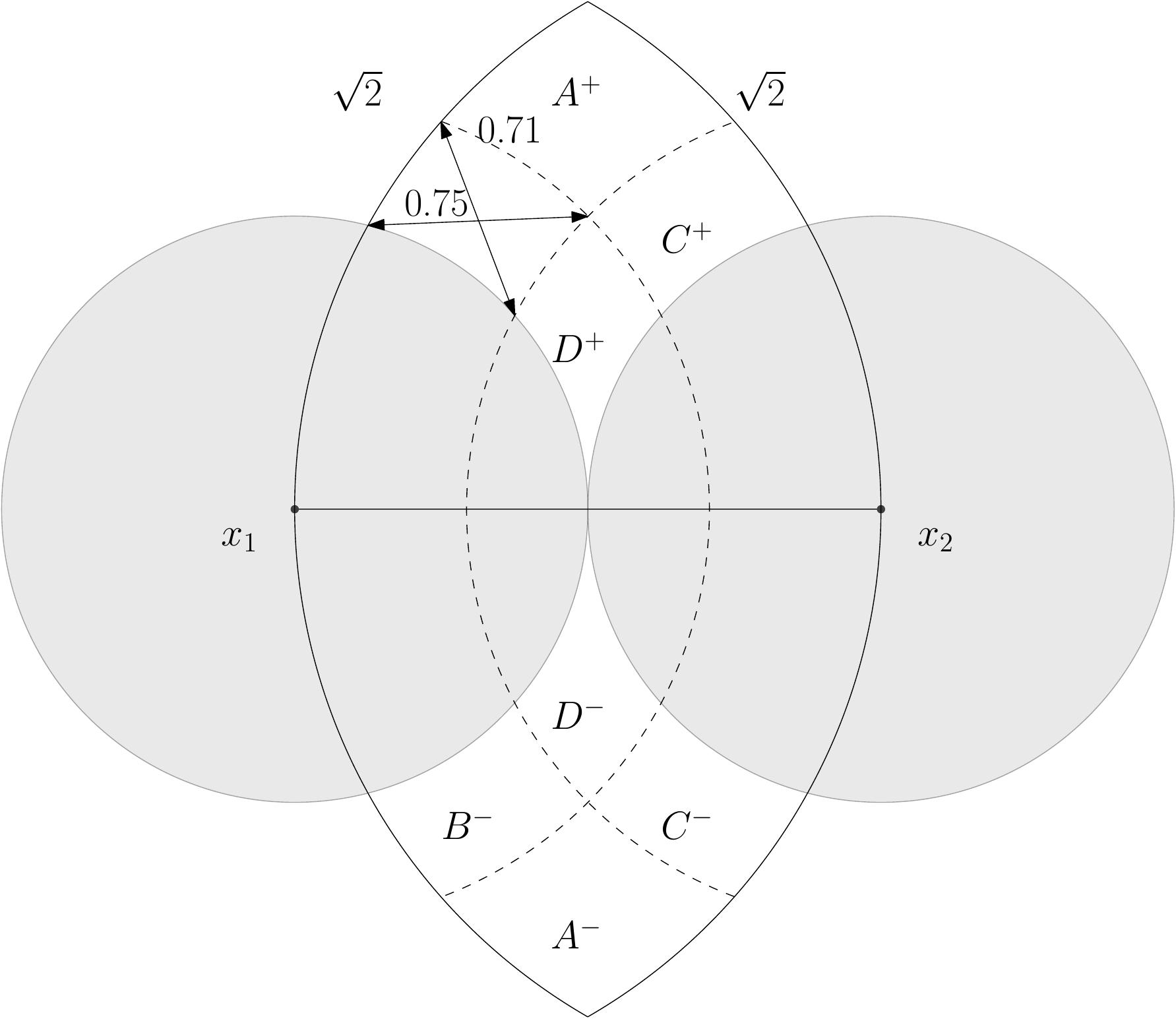}
\caption{The diameter of $B^+$ is 0.75.}\label{fig:diameterB}
\end{minipage}
\begin{minipage}[t][][t]{.5\textwidth}
\includegraphics[width=\textwidth]{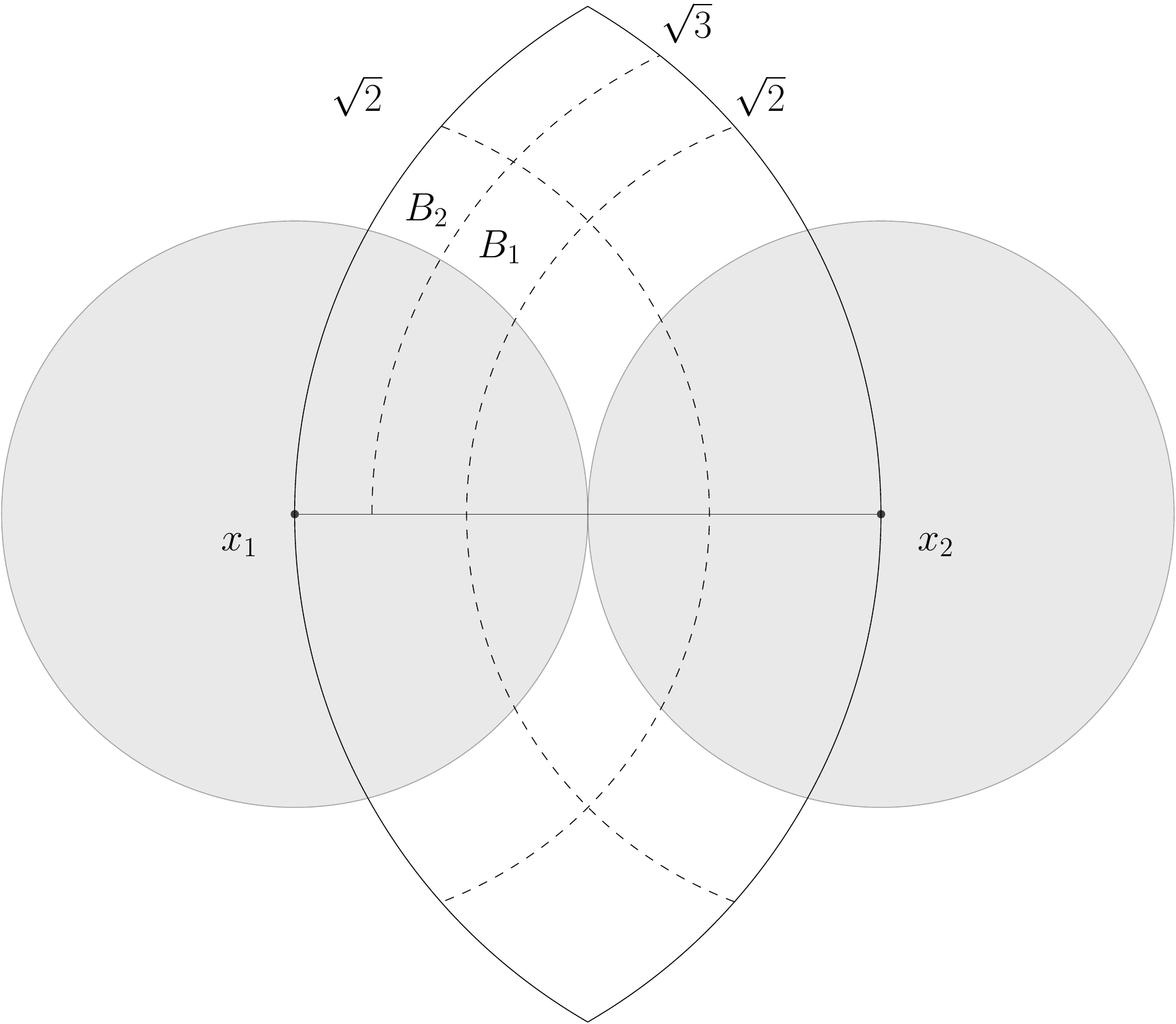}
\caption{Split into $B^+_1$ and $B^+_2$.}\label{fig:splitB1B2}
\end{minipage}
\end{figure}
Any vertex located in $B^+_2$ has a ply disk radius of $\geq \frac{\sqrt{3}}{2} \approx 0.85$. Thus whenever there exists a vertex in $B^+_2$ there cannot be a vertex in $B^+_1$. Furthermore the diameter of $B^+_1$ is 0.54 hence there exists at most one vertex with ply disk radius of $\frac{\sqrt{2}}{2}\approx 0.7$.
\end{proof}

\begin{observation}\label{Obs:A+=1}
There is at most one vertex in $A^+$ and by symmetry in $A^-$.
\end{observation}\setcounter{case}{0}
\begin{proof}
We split 
\[A^+ = \{x \in \R^2 : \sqrt{2} < |x_1x|\leq 2 , \sqrt{2} < |x_2x| \leq 2\}\] 
into four parts as in Figure \ref{fig:splitA}:
\begin{align*}
A_1 &= \{x \in \R^2 : \sqrt{3} < |x_1x|\leq 2, \sqrt{3} < |x_2x| \leq 2\},\\
A_2 &= \{x \in \R^2 : \sqrt{2} < |x_1x|\leq \sqrt{3}, \sqrt{3} < |x_2x| \leq 2\},\\
A_3 &= \{x \in \R^2 :\sqrt{3} < |x_1x|\leq 2, \sqrt{2} < |x_2x| \leq \sqrt{3}\},\\
A_4 &= \{x \in \R^2 : \sqrt{2} < |x_1x|\leq \sqrt{3}, \sqrt{2} < |x_2x| \leq \sqrt{3}\}.
\end{align*}
\begin{figure}
\begin{minipage}[t][][t]{.5\textwidth}
\includegraphics[width=\textwidth]{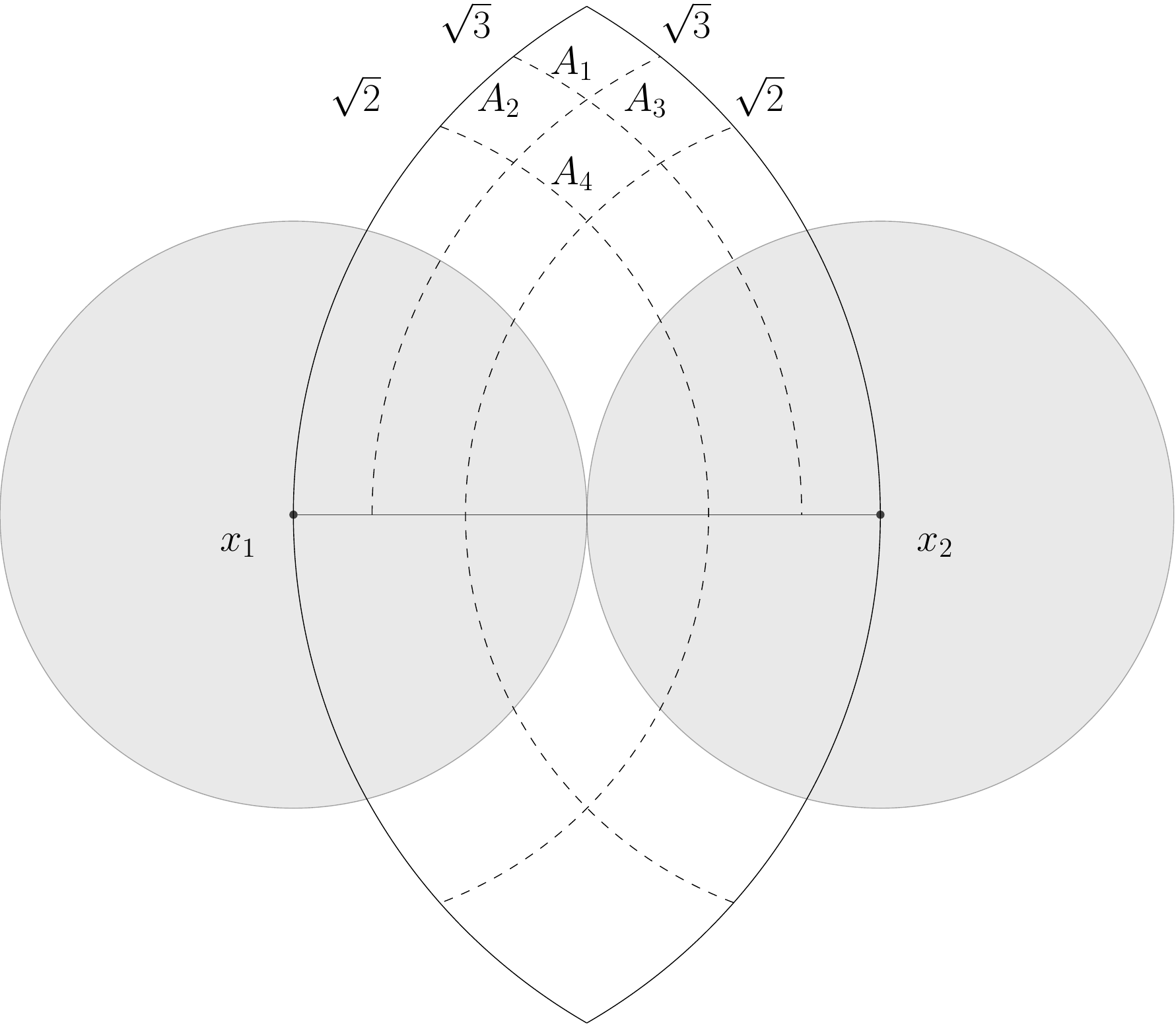}
\caption{Split $A^+$ into four regions.\newline} \label{fig:splitA}
\end{minipage}
\begin{minipage}[t][][t]{.5\textwidth}
\includegraphics[width=\textwidth]{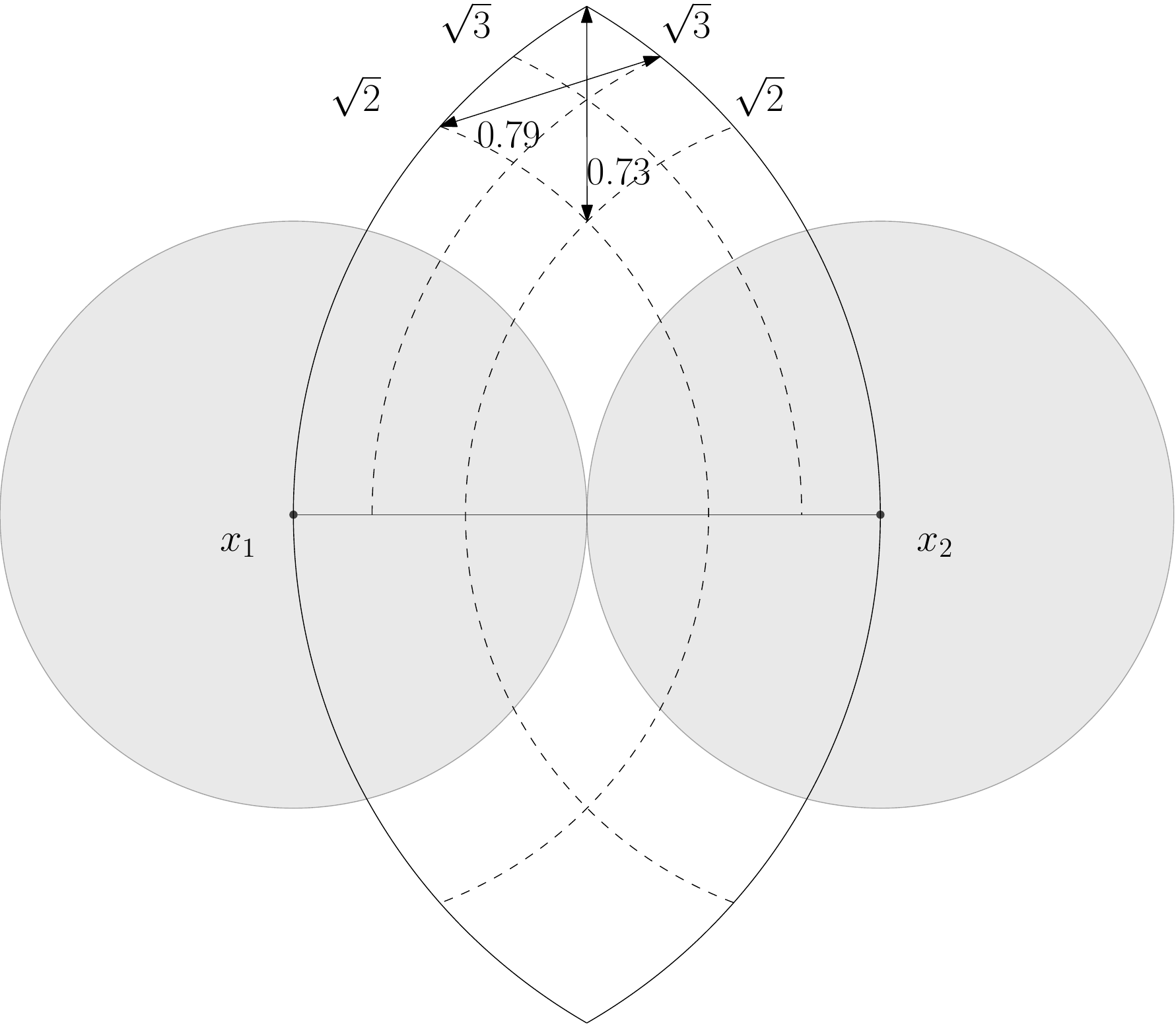}
\caption{Maximal distance of any vertex in $A$ to any vertex in $A_1$.}\label{fig:diameterA+}
\end{minipage}
\end{figure}
\begin{case} A vertex is placed in $A^+_1$.\\
The maximal distance of any vertex in $A^+$ to a vertex in $A^+_1$ is $\leq 0.79$ \mbox{(cf. Figure \ref{fig:diameterA+}),} but the ply disk has a radius of at least $\frac{\sqrt{3}}{2} \approx 0.85$. Thus whenever there exists a vertex in $A^+_1$ there cannot be another vertex in $A^+$.
\end{case}

\begin{case} A vertex is placed in $A^+_2$.\\
Similarly to Case 1. the maximal distance of any two vertices in $A^+_2$ and $A^+_4$ is less than 0.76 and thereby $A^+_4$ has to be empty since the ply disk has a radius of at least $\frac{\sqrt{3}}{2} \approx0.85$ (cf. Figure \ref{fig:diameterA2+}).\\
Assume a second vertex in $A^+_3$. The ply disk of the vertex in $A^+_2$ has a radius of at least $\frac{\sqrt{3}}{2}$ and thus this vertex needs to be placed at a distance of at least $\frac{\sqrt{3}}{2}$ away from the rightmost coordinate of $A^+_3$. Now the ply disk of the vertex in $A^+_2$ is at least $\approx 0.92$. A scheme is presented in Figure \ref{fig:circlesA3}.\\
\begin{figure}
\begin{minipage}[t][][t]{.5\textwidth}
\includegraphics[width=\textwidth]{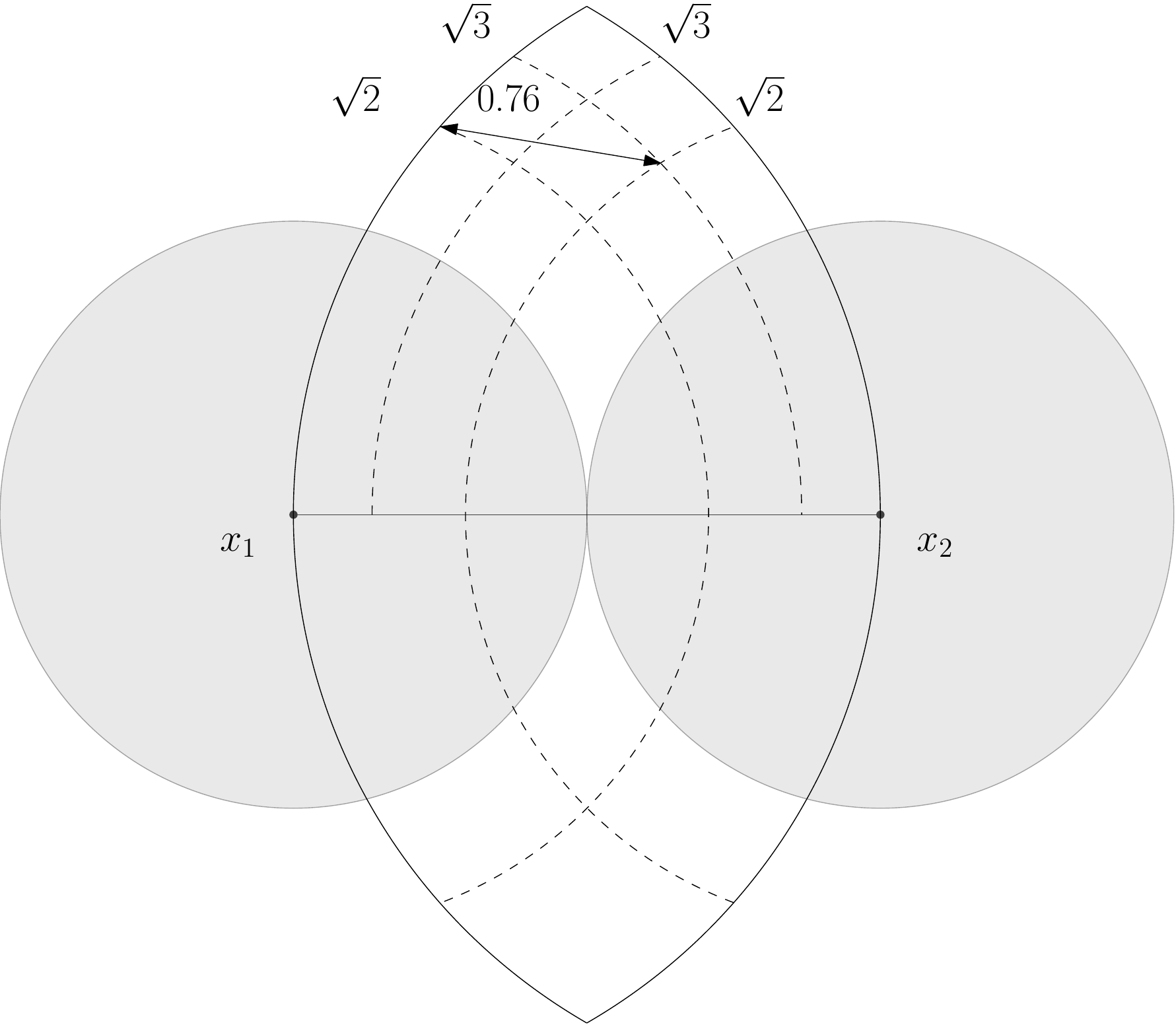} 
\caption{The maximal distance between any two vertices in $A^+_2$ and $A^+_4$}\label{fig:diameterA2+}
\end{minipage}
\begin{minipage}[t][][t]{.5\textwidth}
\includegraphics[width=\textwidth]{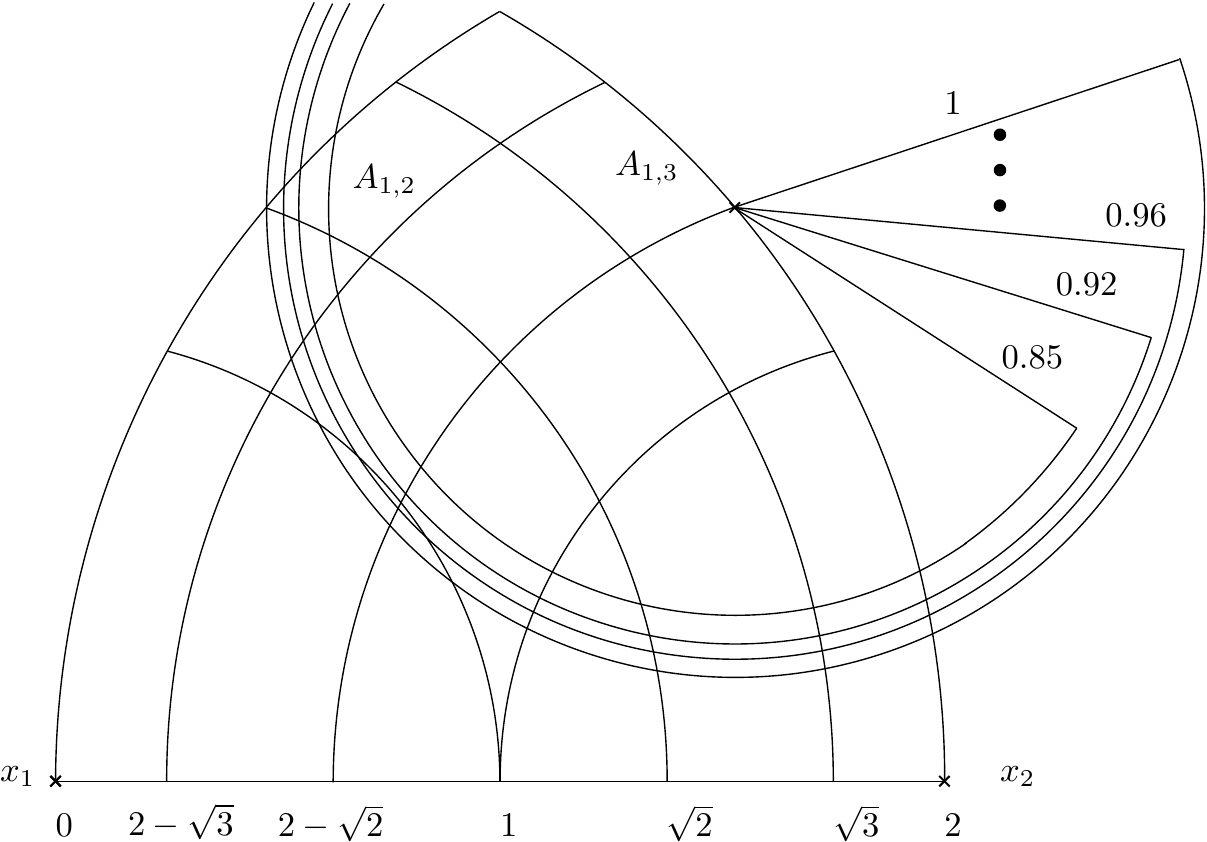}
\caption{A schematic description of the increasing ply-radius of the vertex in $A^+_2$.}\label{fig:circlesA3}
\end{minipage}
\end{figure}
To express this in detail we introduce a series based on the distances: \\
$f(n)$ describes the minimal distance between $x_2$ and the vertex in $A^+_2$ at the $n$th iteration. $f(n)$ is the side of the triangle defined by the edge length of $\sqrt{2}$, 2 and the angle $41.41^\circ + \alpha$. The angle $\alpha$ depends on the ply radius of the vertex in $A^+_2$ of the previous iteration. This is illustrated in Figure \ref{fig:formulaFn}:
\begin{align*}
f(1) &= \sqrt{3}\\
f(n) &= \sqrt{2^2 + \sqrt{2}^2 - 2 \cdot 2 \cdot \sqrt{2} 
\left( \frac{3}{4} \cos(\alpha) - \frac{\sqrt{7}}{4} \cdot \sin(\alpha)\right)}\\
\cos(\alpha) &= \left( \frac{2^2 + \sqrt{2}^2 - \left(\frac{f(n-1)}{2}\right)^2}{2 \cdot 2 \cdot \sqrt{2}} \right)
\end{align*}


We now observe that $\displaystyle \lim_{n \rightarrow \infty} f(n) = 2$. Thereby if we want to place one vertex in $A^+_2$ and one vertex in $A^+_3$ we actually place a vertex at the highest coordinate in $B^+$ and a vertex at the highest coordinate in $C^+$ with distance 1 of each other. We conclude that whenever there exists one vertex in $A^+_2$ then $A^+_3$ has to be empty.

\begin{figure}
\centering
\includegraphics[width=.5\textwidth]{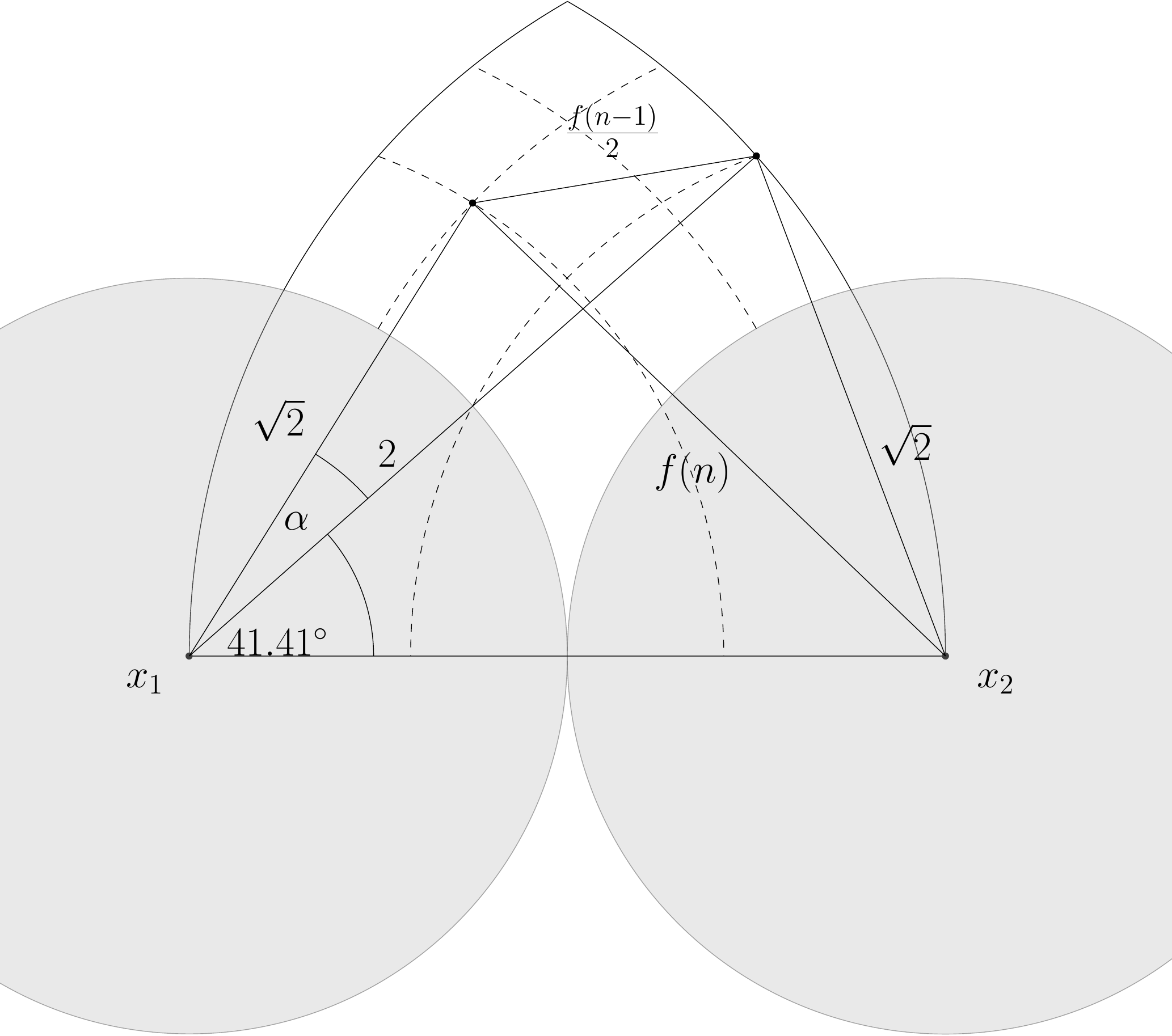}
\caption{$f(n)$ describes the minimal valid distance of any vertex in $A^+_2$ to the vertex $x_2$. $f(n)$ is the length of the side of the triangle defined by the edge length of $\sqrt{2}$, 2 and the angle $41.41^\circ + \alpha$. The angle $\alpha$ depends on the ply radius of the vertex in $A^+_2$ of the previous iteration.}\label{fig:formulaFn}
\end{figure}
\end{case}

\begin{case} A vertex is placed in $A^+_4$.\\
The diameter of the region $A^+_4$ is 0.5. With a ply disk of at least $\frac{\sqrt{2}}{2} \approx 0.7 $ there is at most one vertex in $A^+_4$.
\end{case}

The three cases conclude our observation, that there exists at most one vertex in $A^+$ which can be either in $A^+_1, A^+_2, A^+_3$ or $A^+_4$.
\end{proof}

\begin{observation}\label{Obs:D+=2}
There are at most two vertices in $D^+$.
\end{observation}
\begin{proof}
Let $x \in D^+$ be the point with the largest distance from $e$. Note that this point is unique. The diameter of the set $D^+ \setminus C_x$, where $C_x$ is the ply disk of $x$, is less than $0.3$ and hence there cannot be placed two more disks with radius at least $0.5$ (cf. Figure \ref{fig:diameterD}).
\begin{figure}
\centering
\includegraphics[width=.5\textwidth]{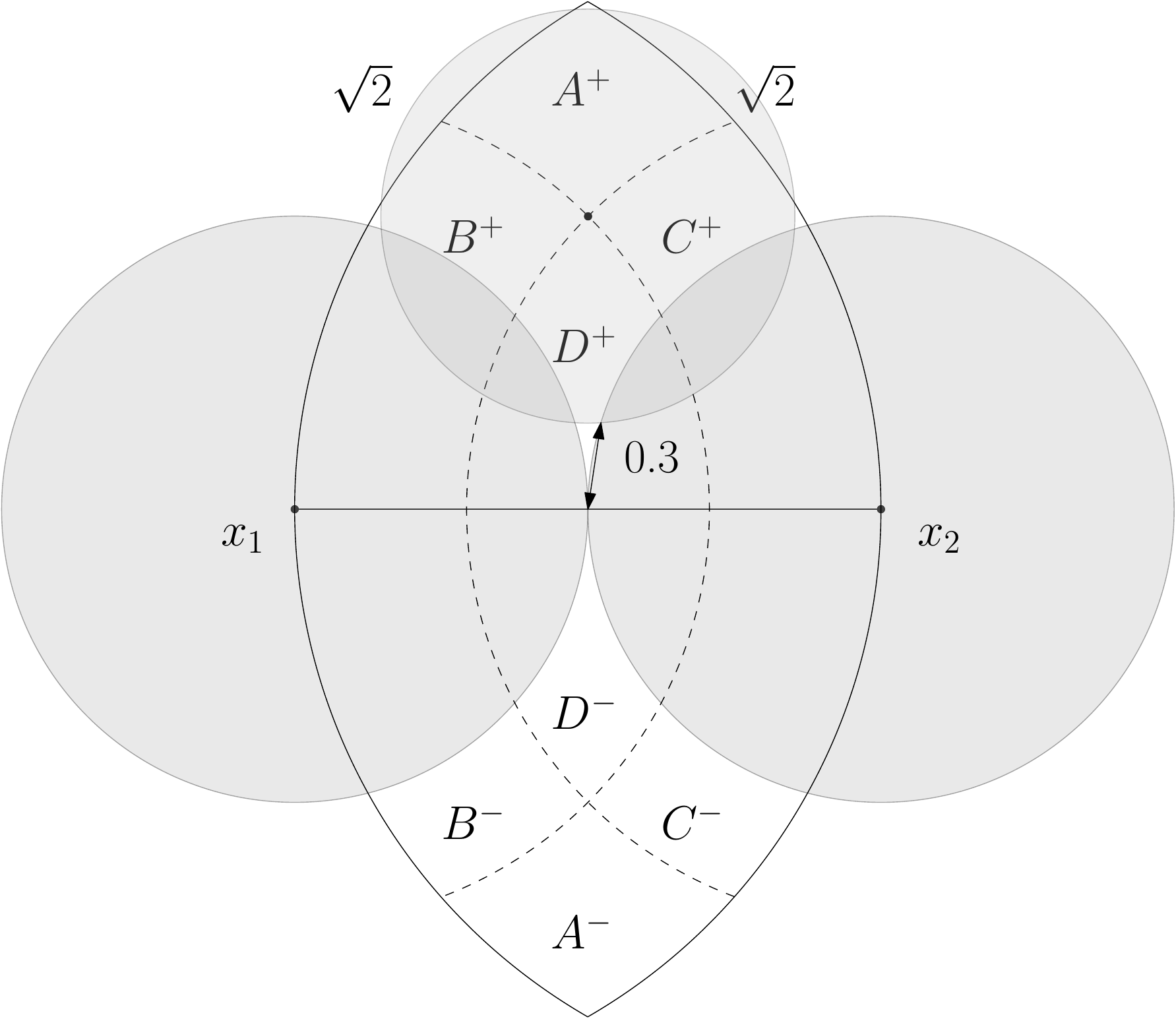}
\caption{One vertex is placed at the highest coordinate in $D^+$. Its ply disk indicates a small remaining region for any other vertices in $D^+$. }\label{fig:diameterD}
\end{figure}

\end{proof}

\begin{observation}\label{Obs:D+D-=3}
There are at most three vertices in $D^+$ and $D^-$.
\end{observation}
\begin{proof}
Consider the top vertex $d_t$ to place in $D^+$. The distance of either $x_1$ or $x_2$ needs to be at most $2 \times \text{dist}(d_t, c)$ where $c$ is the center of $(x_1, x_2)$. Solving the equation 
\[\text{dist}(x, d_t) = 2 \times \text{dist}(c, d_t)\]
for $x = x_1$ and $x_2$ partitions the region $D^+$ into a top ($D^+_1$) and a bottom part.
Placing one vertex in $D^+_1$ and one in $D^-_1$ we can bound the region to place the second vertex in $D^+$ to have a minimal distance of $\frac{1.58}{2}$ of the topmost coordinate in $D^+_1$ (placing a vertex at the topmost coordinate of $D^-_3$). The region between $c$ and this bound is called $D^+_3$ and the region in between $D^+_2$.
Any vertex in $D^+_2$ excludes vertices in $D^+_1$ and $D^+_3$ and thus of any two vertices in $D^+$ $d_t$ has to lie in $D^+_1$ where $d_b$ (the bottom vertex) has to lie in $D^+_3$. The regions are drawn in Figure \ref{fig:splitD1D2D3}.
\begin{figure}
\centering
\includegraphics[width=.5\textwidth]{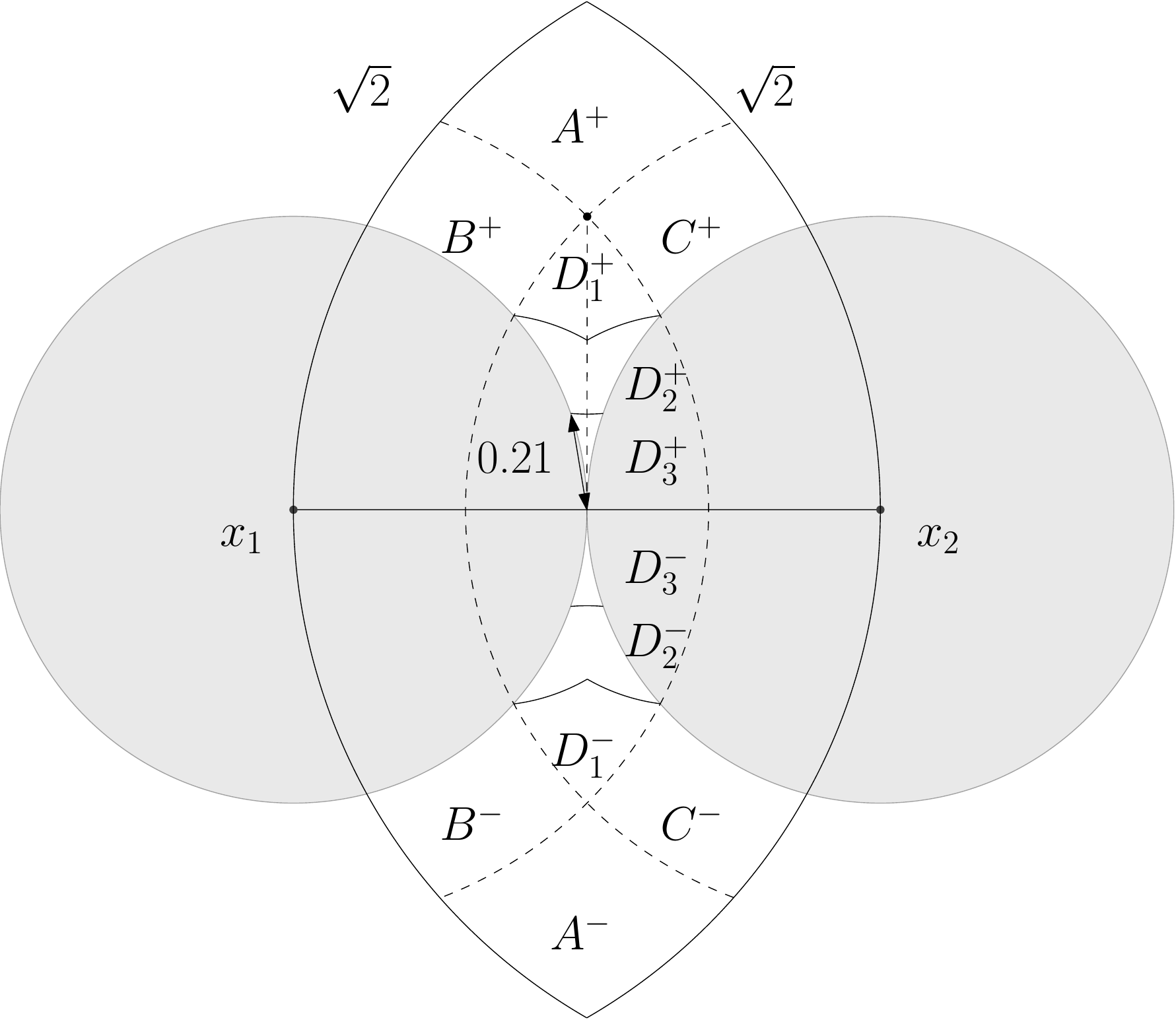}
\caption{Partition of $D^+$ and $D^-$. Whenever there is a vertex in $D^+_2$ then $D^+_1$ and $D^+_3$ are empty.}\label{fig:splitD1D2D3}
\end{figure}

Any four vertices in $|D^+| + |D^-|$ have to be placed in $D^+_1, D^+_3, D^-_1$ and $D^-_3$. Two vertices in $D^+_3$ and $D^-_3$ need a distance $\geq 0.5$. Thus according to the diameter of $D^+_3 \cup D^-_3 \leq 0.42$ there can not be two vertices with ply disk radii of at least $ 0.5$.

\end{proof}

\begin{observation}\label{Obs:A^+=1A-=0}
Whenever there exists a vertex in $A^+$ there cannot be a vertex in $A^-$.
\end{observation}
\begin{proof}
Assume there exist one vertex in $A^+$ and one vertex in $A^-$. The minimal distance these vertices can have is $> 2$ which contradicts the maximality of the edge $(x_1, x_2)$.
\end{proof}

\begin{observation}\label{Obs:A+=1D-=1}
Whenever one vertex exists in $A^-$ there is at most one vertex in $D^+$.
\end{observation}
\begin{proof}
Assume the vertex $a \in A^-$ is fixed at the highest possible position. The distance of $a$ to the center $c$ of $(x_1, x_2)$ is $> 1$. Where as the distance of any point $d \in D^+$ to $c$ is $<1$ due to the maximality of $(x_1, x_2)$. 
Thus if $d$ is at the highest possible coordinate $c$ is covered by its ply disk (distance $a$ to $b$ = 2).\\
Solving the equation 
\[\text{dist}(a, c) = 2 \text{dist}(x,c)\]
 results in a circle which includes $D^+$ completely (shown in Figure \ref{OnevertexInA-}). The circle is centered $\frac{1}{3}$ above $c$ and has a radius of $\frac{2}{3}$. The ply disk of any point $d$ in the circle covers $c$ and thus there can not be a second vertex in $D^+$.
\begin{figure}
\centering
\includegraphics[width=.5\textwidth]{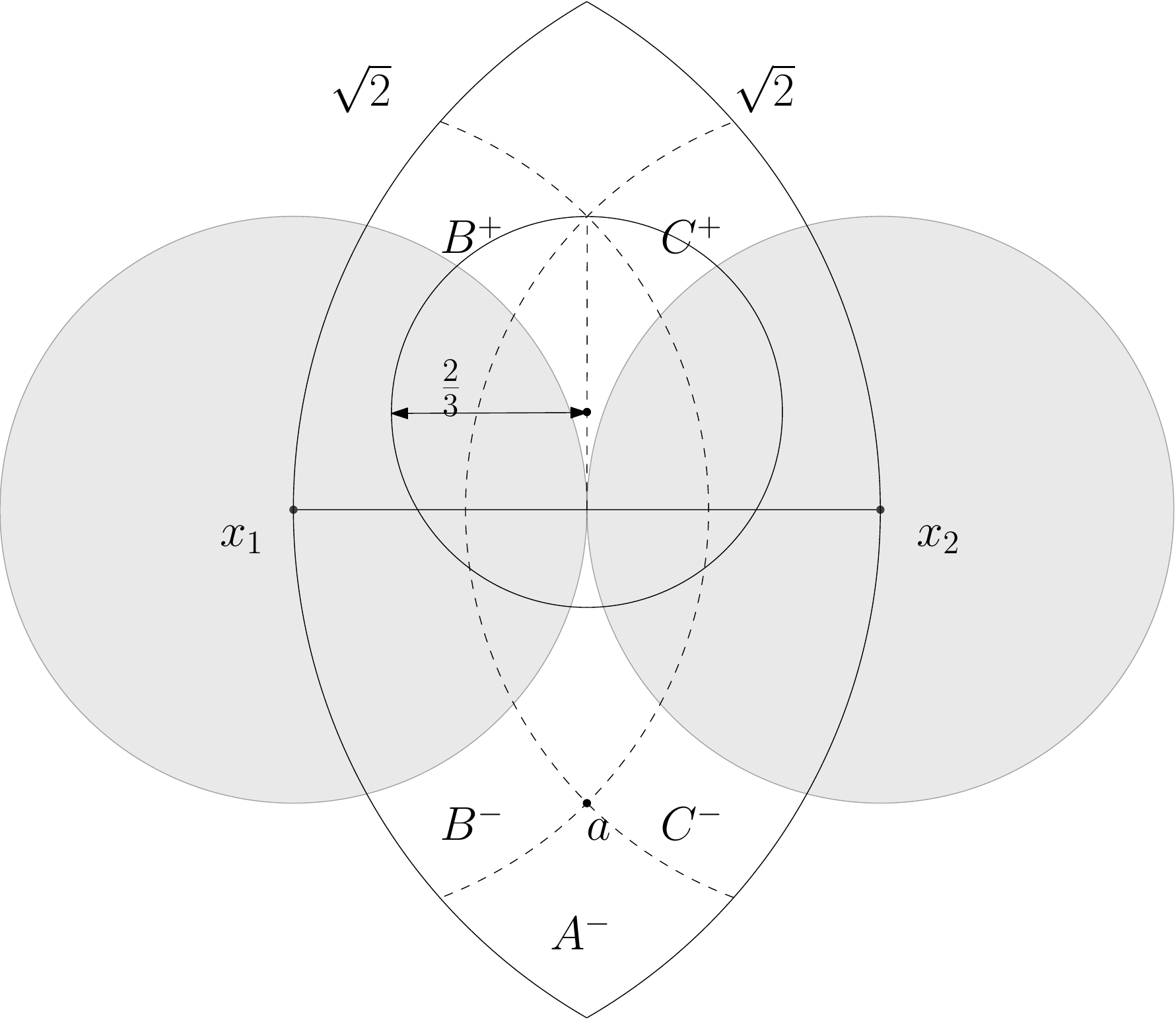}
\caption{A vertex in $A^-$ causes the ply disk of any vertex in $D^+$ to cover the center of $(x_1, x_2)$ and thus there can not be a second vertex in $D^+$. }\label{OnevertexInA-}
\end{figure}
\end{proof}

\begin{observation}\label{Obs:A+=1D+D-=2}
Whenever there exists one vertex in $A^+$ then there are at most two vertices in $D^+ \cup D^-$.
\end{observation}
\begin{proof}
Assume there exist four vertices in $A^+, D^+$ and $D^-$. By Observation \ref{Obs:A+=1D-=1} there is at most one vertex in $D^-$ thus there are exactly two vertices in $D^+$.

Note that by Observation \ref{Obs:D+D-=3} the top vertex in $D^+$ has at least the distance $\frac{1}{\sqrt{3}}$ from $c$. Assume the top vertex in $D^+$ to be placed exactly at this coordinate.
The minimal distance $d$ from the vertex in $A^+$ to either $x_1$ or $x_2$ can be calculated by solving the following equation:
\begin{align*}
d^2 &= 1^2 + \left(\frac{d}{2} + \frac{1}{\sqrt{3}}\right)^2\\
d^2 &= 1 + \left(\frac{d}{2}\right)^2 + 2 \cdot \frac{d}{2} \cdot \frac{1}{\sqrt{3}} + \frac{1}{3}\\
d^2 ( 1- \frac{1}{4}) &= 1 + \frac{1}{3} + \frac{d}{\sqrt{3}}\\
\frac{3}{4} d^2 - \frac{1}{\sqrt{3}} d - \frac{4}{3} &= 0
\end{align*}

Which has the one positive solution $d = \frac{2}{9} \left(\sqrt{3} + \sqrt{39}\right) \approx 1.77$ and thus the minimal distance of any vertex in $A^+$ to the center of $(x_1, x_3)$ is $\frac{d}{2} + \frac{1}{\sqrt{3}} \approx  1.47$ (cf. Figure \ref{fig:OnevertexBelow}).
The distance of any vertex in $D^-$ to the vertex in $A^+$ has to be $\leq 2$ due to the maximality of the edge $(x_1, x_2)$.
But the ply disk of the second vertex in $D^+$ covers the remaining region to place any vertex below $(x_1, x_2)$. 
This contradicts our assumption that there exist four vertices in $A^+, D^+$ and $D^-$. 
\end{proof}
\begin{figure}
\centering
\includegraphics[width=.40\textwidth]{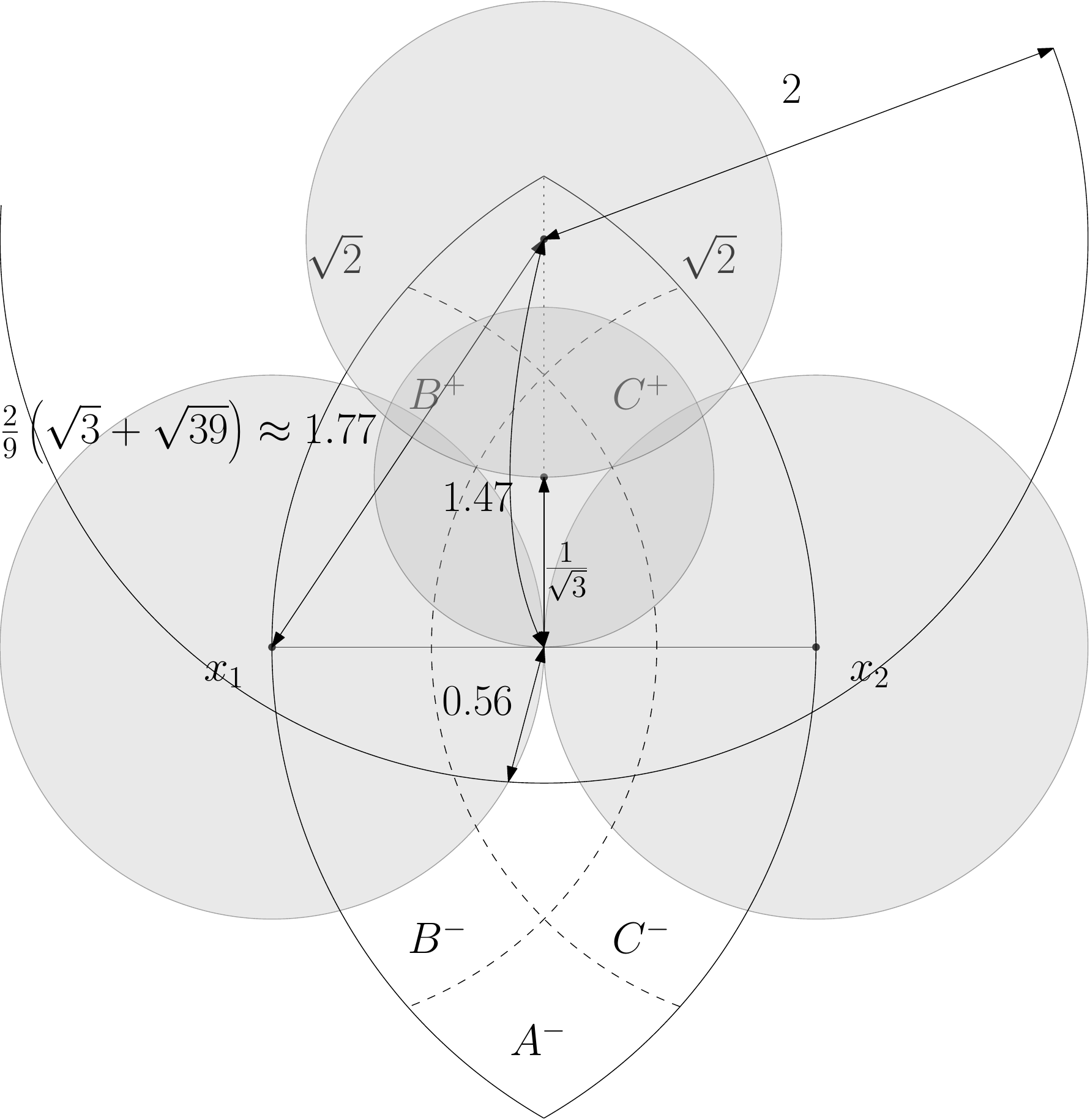}
\caption{Whenever there exist exactly two vertices in $D^+$ and one in $A^+$ there cannot exist a vertex in $D^-$.}\label{fig:OnevertexBelow}
\end{figure}

\begin{observation}\label{Obs:B+C-=2D+D-=1}
Whenever there exists exactly one vertex in $B^-$ and exactly one vertex in $C^+$ then there exists at most one vertex in $D^+$ or $D^-$.
\end{observation}\setcounter{case}{0}

\begin{proof}
We will use the same partition of $B^-$ and $C^+$ as in Observation \ref{Obs:B+=1} and distinguish three cases about the placement within $B^-$ and $C^+$.
\begin{case}
There exist one vertex in $B^-_2$ and one vertex in $C^+_2$.
\end{case}

Note that the placement of the vertices is unique in this case. The distance of the two vertices is 2. Their ply disks meet at the center of $(x_1, x_2)$ and the rest of $D^+$ and $D^-$ is covered. This implies a unique coordinate to place a vertex in $D^+$.

\begin{case}
There exist one vertex in $B^-_1$ and one vertex in $C^+_2$.
\end{case}
Note that the maximal diameter of any possible region to place a vertex in $D^+$ is $\leq 0.59$. This region can be obtained by placing the vertex in $C^+_2$ at the rightmost coordinate and the vertex in $B^-_1$ at the lowest possible coordinate. 
Observe that moving any of both vertices reduce the diameter of the region, since either the topmost coordinate of the region moves down or the bottommost coordinate of the region moves up (cf. Figure \ref{fig:longDiagonal}).\\
The distance from $x_2$ to the center between the vertex in $C^+_2$ and $x_1$ is $\sqrt{\frac{3}{2}}$ and thus any vertex in the top region has a ply radius of $\leq 0.61$ any thus covers the remaining region of $D^+$ completely. The region below has a diameter of $\leq 0.34$ and therefore there exists at most one vertex with ply radius $\geq 0.5$. 
\begin{figure}
\centering
\includegraphics[width=.5\textwidth]{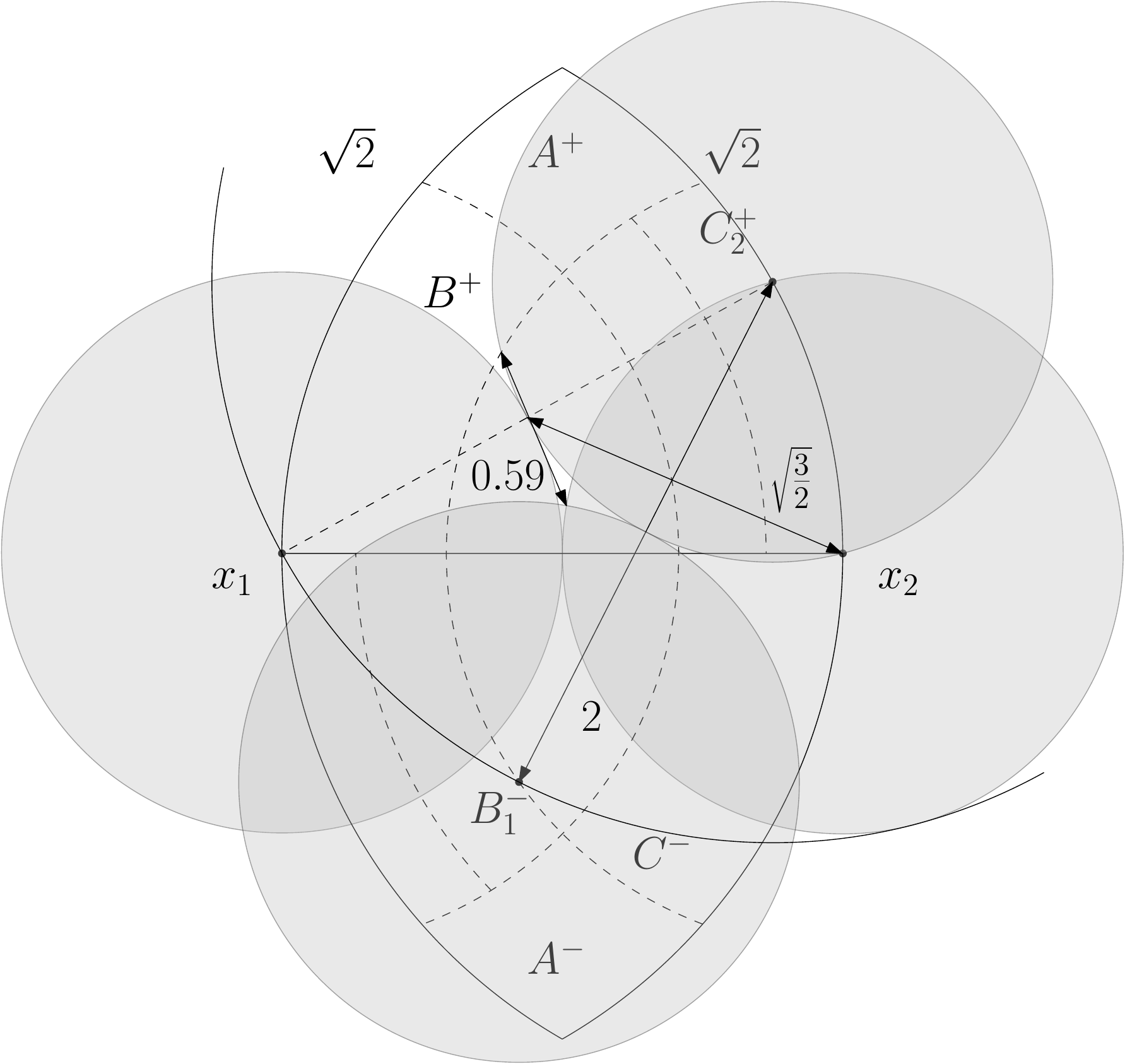}
\caption{The region with the maximal possible diameter in $D^+$ whenever there exist vertices in $B^-_1$ and $C^+_2$ }\label{fig:longDiagonal}
\end{figure}

\begin{case}
There exist one vertex in $B^-_1$ and one vertex in $C^+_1$.
\end{case}
Note that in any configuration the rightmost coordinate of $D^-$ and the leftmost coordinate of $D^+$ is covered by the ply disks of the vertex in $B^-_1$ and of the vertex in $C^+_1$. Thus $D_3^{+-}$ is covered in any case. All possible configurations result in one of the following cases:
\begin{enumerate}
\item[a)] $D^+_2$ and $D^-_2$ are covered.\\ Similar to Observation \ref{Obs:D+D-=3} there cannot be two vertices in $D^+_3 \cup D^-_3$.
\item[b)] Exactly one of the $D_2$ is not covered completely and the center of $(x_1, x_2)$ is covered.\\ Thus there can be at most one vertex in either $D^+$ or $D^-$.
\end{enumerate}
\end{proof}
This concludes Observation \ref{Obs:B+C-=2D+D-=1} saying whenever there exist vertices in $B^-$ and $C^+$ there exists at most one vertex in $D^+ \cup D^-$.

\bigskip
\begin{observation}\label{Obs:B-A+=2B+C+=1}
Whenever there exist one vertex in $B^-$ and one vertex in $A^+$ either $B^+$ or $C^+$ is empty.
\end{observation}

\begin{proof}
Any vertex in $A^+$ has a ply radius of at least $\frac{\text{dist}(B^-, A^+)}{2} \approx \frac{1.68}{2} $.Thus this vertex covers either $B^+$ or $C^+$ completely whether the vertex is closer to $B^+$ or $C^+$ as presented in Figure \ref{fig:B-andA+}.
\begin{figure}
\centering
\includegraphics[width=.5\textwidth]{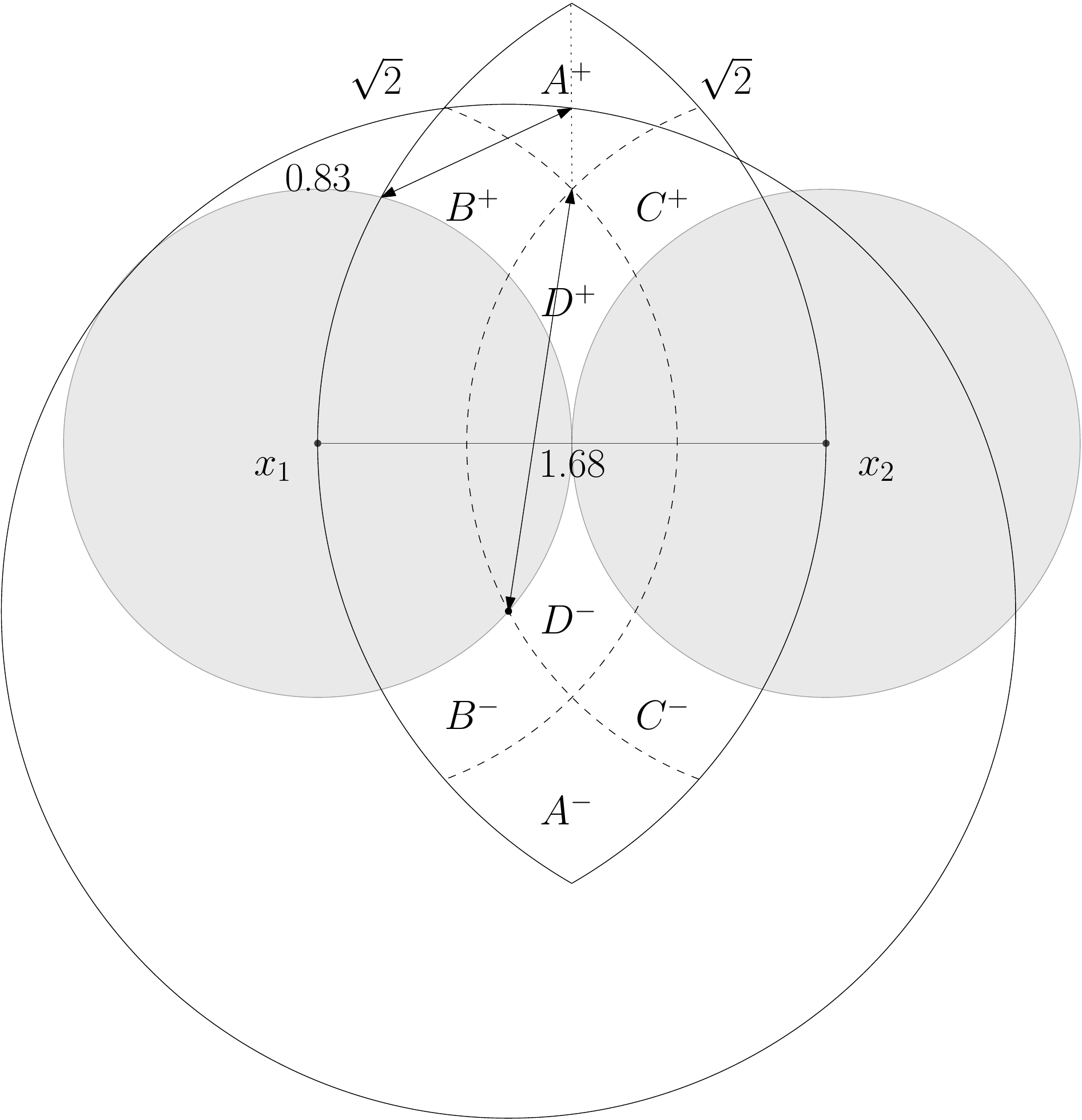}
\caption{Any vertex in $A^+$ has a ply radius of at least $\frac{\text{dist}(B^-, A^+)}{2} \approx \frac{1.68}{2} $.Thus this vertex covers either $B^+$ or $C^+$ completely whether the vertex is closer to $B^+$ or $C^+$.}\label{fig:B-andA+}
\end{figure}
\end{proof}
In the next lemmas we distinguish the different cases about the number of vertices placed below the edge $(x_1, x_2)$.

\begin{lemma}\label{noVertexAbove}
Whenever no vertex is below the edge $(x_1, x_2)$ then $n = 2 + |A^+| + |B^+| + |C^+| + |D^+| \leq 7$.
\end{lemma}
\begin{proof}
By Observation \ref{Obs:B+=1}, $|B^+| \leq 1$ and $|C^+| \leq 1$, by Observation \ref{Obs:A+=1}, $|A^+|\leq 1$ and by Observation \ref{Obs:D+=2} $|D^+| \leq 2$.
Thereby $|A^+| + |B^+| + |C^+| + |D^+| \leq 5$ and the lemma follows. 
\end{proof}

\begin{lemma}\label{oneVertexBelow}
Whenever exactly one vertex is below the edge $(x_1, x_2)$ then $ n = 3 + |A^+| + |B^+| + |C^+| + |D^+| \leq 7$.
\end{lemma}

\begin{proof}
Assume there exists an empty ply drawing of $K_8$ with exactly one vertex below $(x_1, x_2)$.
By Lemma \ref{noVertexAbove} there have to be exactly five vertices above $(x_1, x_2)$, namely exactly one in $A^+$, one in $B^+$, one in $C^+$ and two in $D^+$.

By Observation \ref{Obs:A^+=1A-=0} the vertex below cannot be placed in $A^-$, by Observation \ref{Obs:B+C-=2D+D-=1} the vertex below cannot be placed in $B^-$ or $C^-$ and by Observation \ref{Obs:A+=1D+D-=2} the vertex below cannot be placed in $D^-$. 
This contradicts our assumption that $K_n$ with $n\geq 8$ can be drawn with one vertex below $(x_1, x_2)$. 
\end{proof}

 \begin{lemma}
 Whenever exactly two vertices are below $(x_1,x_2)$ then $n = 4 + |A^+| + |B^+| + |C^+| + |D^+| \leq 7$.
 \end{lemma} \setcounter{case}{0}
 
\begin{proof}
 Assume there exists an empty ply drawing of $K_n$ with $n \geq 8$ there are at least four vertices above $(x_1, x_2)$. We will distinguish four cases about the placement of the two vertices below.
 
\begin{case}
 There exists one vertex in $A^-$.
 \end{case}
By Observation \ref{Obs:B+=1}, $|B^+| \leq 1$ and $|C^+|\leq 1$,  by Observation \ref{Obs:A^+=1A-=0}, $|A^+| = 0$ and by Observation \ref{Obs:A+=1D-=1}, $|D^+| \leq 1$.
 
 Thus whenever there exists one vertex in $A^-$, $n = 4 +  |A^+| + |B^+| + |C^+| + |D^+| \leq 7$ holds.
  
  \begin{case}
$A^-$ is empty and there are exactly two vertices in $D^-$.
 \end{case}
By Observation \ref{Obs:B+=1}, $|B^+| \leq 1$ and $|C^+|\leq 1$,  by Observation \ref{Obs:A+=1D-=1}, $|A^+| = 0$, and by Observation \ref{Obs:D+D-=3}, $|D^+| \leq 1$.

Thus whenever there exists two vertices in $D^-$, $n = 4 +  |A^+| + |B^+| + |C^+| + |D^+| \leq 7$ holds. 
  
 \begin{case}
 $A^-$ is empty and there is exactly one vertex in $D^-$.
 \end{case}
 There exists a vertex in either $B^-$ or in $C^-$. Since both cases are symmetric we assume without loss of generality there exists exactly one vertex in $B^-$.
 
\begin{enumerate}
\item[a)] Assume there exists a vertex in $A^+$:\\
By Observation \ref{Obs:B-A+=2B+C+=1} there exists at most one vertex in either $B^+$ or $C^+$.
Furthermore since there exists a vertex in $D^-$ we know by Observation \ref{Obs:A+=1D+D-=2} that there exists at most one vertex in $D^+$.
Thereby $n = 4 + |A^+| + |B^+| + |C^+| + |D^+| \leq 7$ holds.

\item[b)] Assume $A^+$ to be empty:\\
Assuming $K_8$ is empty ply drawable and there exist four vertices above, namely exactly one in $B^+$, one in $C^+$ and two in $D^+$. By Observetion \ref{Obs:B+C-=2D+D-=1} there can be at most one vertex in $D^+ \cup D^-$ which is already placed in $D^-$. 
\end{enumerate}
This concludes the case where since $n = 4 + |A^+| + |B^+| + |C^+| + |D^+| \leq 7$ holds in both subcases. 

 \begin{case}
 There exist exactly one vertex in $B^-$ and exactly one vertex in $C^-$.
 \end{case}
 Assuming $K_8$ is empty ply drawable and there exist four vertices above. 
 By Observation \ref{Obs:B+=1}, \ref{Obs:A+=1} and \ref{Obs:D+=2} there exists at least one vertex in either $B^+$ or $C^+$. By Observetion \ref{Obs:B+C-=2D+D-=1} there can be at most one vertex in $D^+ \cup D^-$. Hence $A^+$ must also contain one vertex. 
	By Observation \ref{Obs:B-A+=2B+C+=1} there exists at most one vertex in either $B^+$ or $C^+$. Thereby $n = 4 + |A^+| + |B^+| + |C^+| + |D^+| \leq 7$ holds . 
 \end{proof}
 
 \begin{lemma}
  Whenever exactly three vertices are below $(x_1,x_2)$ then $n = 5 + |A^+| + |B^+| + |C^+| + |D^+| \leq 7$.
 \end{lemma}
 \setcounter{case}{0}
  \begin{proof}
 Note that in any case there exists at least one vertex in either $B^-$ or $C^-$.
 We distinguish the following cases by the placement of the vertices below $(x_1, x_2)$. 
 \begin{case}
	There exist three vertices below, namely in $A^-$, in $B^-$ and in $C^-$.
 \end{case}
 By Observation \ref{Obs:A^+=1A-=0}, $|A^+| = 0$.
 $B^+$ and $C^+$ are empty by Observation \ref{Obs:B-A+=2B+C+=1} since there exist vertices in $B^-$ and $C^-$. Additionally there is at most one vertex in $D^+$ by Observation \ref{Obs:A+=1D-=1}. Thereby $n = 5 + |A^+| + |B^+| + |C^+| + |D^+| \leq 7$ holds. 
	
 \begin{case}
	There exist three vertices below, namely in $A^-$, in $D^-$ and without loss of generality in $B^-$.
 \end{case}
  By Observation \ref{Obs:A^+=1A-=0}, $|A^+| = 0$, by Observation \ref{Obs:A+=1D-=1} there exists at most one vertex in $D^+$.\\
  Assuming that $K_8$ is empty ply drawable there have to be three vertices above namely exactly one in $B^+$, exactly one in $C^+$ and exactly one in $D^+$. By Observation \ref{Obs:B+C-=2D+D-=1} there exists a long diagonal and thus $|D^+ \cup D^-| \leq 1$ which contradicts the case that there already exists one vertex in $D^-$.

 \begin{case}
	There exist exactly three vertices below, namely in $B^-$, in $C^-$ and in $D^-$.
 \end{case}
 By Observation \ref{Obs:B+C-=2D+D-=1} whenever there exist a vertex in $B^+$ or $C^+$, $D^+$ is empty since there already exists a vertex in $D^-$.
 
	By Observation \ref{Obs:B-A+=2B+C+=1} whenever there exists a vertex in $A^+$ either $C^+$ or $B^+$ is empty. 
	The remaining possible placements are: 
	\begin{enumerate}
	\item two vertices in $D^+$
	\item one vertex in $B^+$ and
		\begin{enumerate}
			\item one vertex in $A^+$
			\item one vertex in $C^+$
		\end{enumerate}
	\item one vertex in $A^+$ and exactly one vertex in $D^+$
	\end{enumerate}
 Thereby there exist at most two vertices above $(x_1, x_2)$ and $n = 5 + |A^+| + |B^+| + |C^+| + |D^+| \leq 7$.

\begin{case}
	There exist three vertices below, namely two vertices in $D^-$ and without loss of generality one vertex in $B^-$.
 \end{case}
 By Observation \ref{Obs:A+=1D-=1} $|A^+| =0$, since there already exist two vertices in $D^-$.\\
By Observation \ref{Obs:B+C-=2D+D-=1} $|C^+| =0$, since there already exist two vertices in $D^-$.\\
By Observation \ref{Obs:D+D-=3} there is at most one more vertex in $D^+$.\\
By Observation \ref{Obs:B+=1} there is at most one more vertex in $B^+$. \\
Thus there exist at most two vertices above $(x_1, x_2)$ and $n = 5 + |A^+| + |B^+| + |C^+| + |D^+| \leq 7$.

\end{proof}
The four lemmas conclude the proof of Theorem \ref{thm:k8} since the placement of more than 3 vertices below will be symmetric to placing the corresponding number of vertices above.

\section*{Omitted proof of Theorem~\ref{thm:bipartiteK2X}}

In this section we give the omitted proof of Theorem~\ref{thm:bipartiteK2X}:

\medskip
\rephrase{Theorem}{\ref{thm:bipartiteK2X}}{
		Graph $K_{2,m}$ with $m \geq 15$ does not admit any empty-ply drawing.}

\medskip

Before stating the proof, we give some conditions for the vertices of an empty ply drawing of a complete bipartite graph $K_{1,X}$ and then we extend these requirements to $K_{2,X}$, where $X \geq 1$.

Given a complete bipartite graph, $K_{1,X}$ with $X>1$, let $u$ be the only vertex of the first set, called $V_1$, and let $2m_u$ be the distance from $u$ to its farthest vertex (and so the radius of the ply disk of $u$). The $X$ vertices of the second set, called $V_2$, of an empty-ply drawing of $K_{1,X}$ have distance in the range $R = [m_u, 2m_u]$ from $u$. Following the proof of Theorem~\ref{th:empty-ply-max-degree} we split $R$ in two ranges $R_{u_1} = [m_u, \sqrt{2}m_u]$ and $R_{u_2} = [\sqrt{2}m_u, 2m_u]$. 

A condition for the vertices is given by the following lemma.
\begin{lemma}\label{lem:sector}
In an empty ply drawing of $K_{1,X}$ with $X>1$ the angular distance of any two vertices $x_1, x_2 \in V_2$ drawn both in   $R_{u_1}$, or $R_{u_2}$, is $\geq 27.89^\circ$.
\end{lemma}

\begin{figure}
\centering
\includegraphics[width=.5\textwidth]{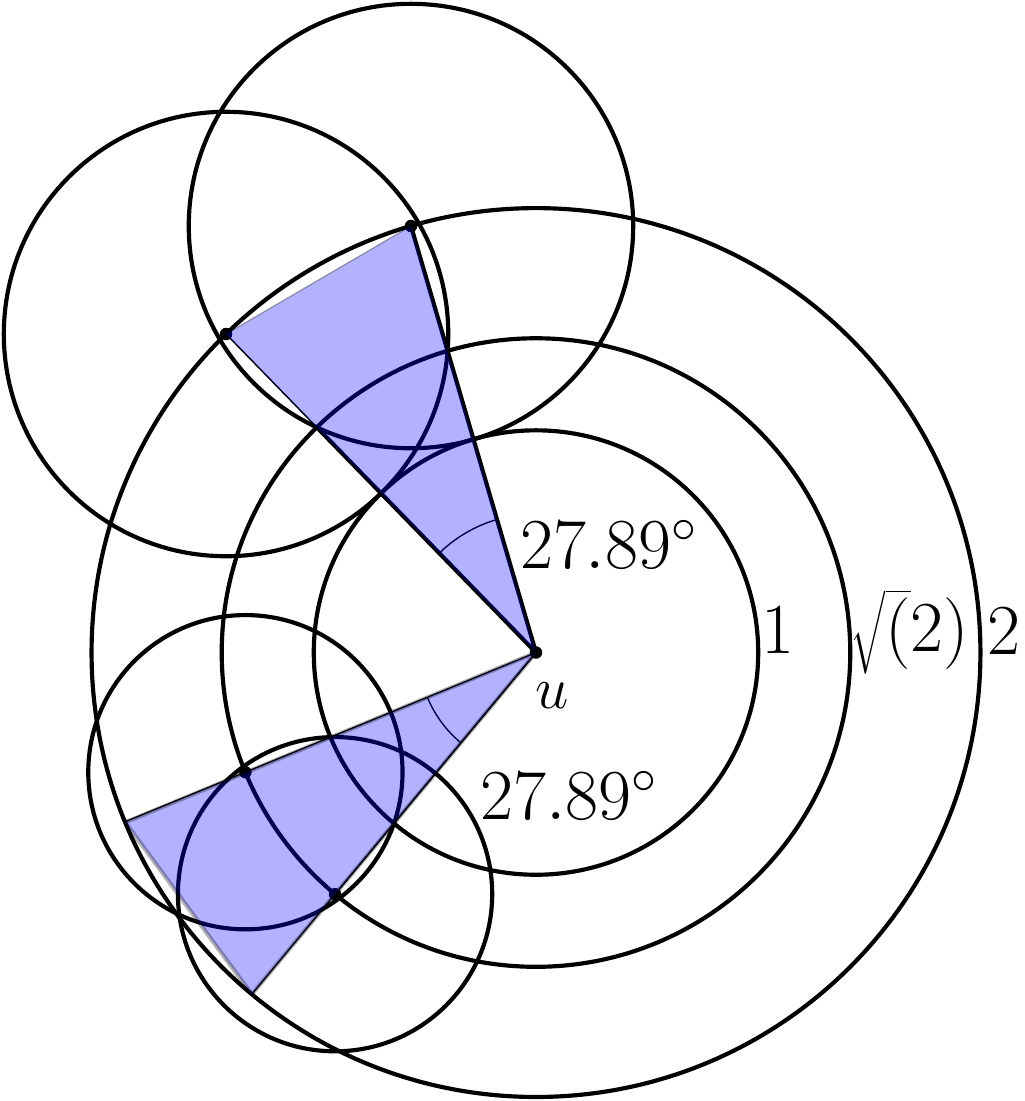}
\caption{The minimum angular distance with respect to $u$ of two vertices both with distance in the same range $[\sqrt{2} m, 2m]$ or $[m, \sqrt{2} m]$ from $u$ is $27.89^\circ$, as shown by the blue circular sector $S$. Thereby in any circular sector with ange $< 27.89^\circ$ there can be at most 2 vertices iff exactly one is in $[\sqrt{2} m, 2m]$ and exactly one is in $[m, \sqrt{2} m]$.}
\label{figAppendix:27degree}
\end{figure}

\begin{proof}
We prove the statement for $R_{u_2}$. The proof for $R_{u_1}$ is analogous.\
Given two vertices $x_1, x_2 \in V_2$ in $R_{u_2}$ drawn at distance $2m_u$ and $\sqrt{2}m_u$ from $u$ respectively, (i.e. the maximum and the minimum distances from $u$) the triangle formed by $u$, $x_1$ and $x_2$ has sides with length $2m_u$, $m_u$ and $\sqrt{2}m_u$. Thus the angle formed by the two sides incident to $u$ is given by the following equation:
\[1^2 \cdot m_u^2 = (2^2 + \sqrt{2}^2 - 2 \cdot 2 \cdot \sqrt{2} \cdot \cos(\alpha_d)) \cdot m_u^2\]
\[\cos(\alpha_d) = \frac{2^2 + \sqrt{2}^2 - 1^2}{ 2 \cdot 2 \cdot \sqrt{2}}\]
\[\alpha_d = 27.89^\circ \]
It follows that two vertices drawn both in $R_{u_2}$ (resp. $R_{u_1}$) and having angular distance $\alpha_d$ with respect to $u$ have distance  $\sqrt{2}m_u$ and $2m_u$ (resp. $m_u$ and $\sqrt{m_u}$) from $u$.
\end{proof}

Let $S \cap R_{u_1}$ and $S \cap R_{u_2}$ be the ranges of $S$, Lemma~\ref{lem:sector} implies that a circular sectors $S$ around $u$ with angle $\alpha = (\alpha_d -\epsilon)^\circ $ contains at most one vertex per range.

We now extend the property given in Lemma~\ref{lem:sector} to the case when $V_1$ contains two vertices, that is $K_{2,X}$.

Given a complete bipartite graph, $K_{2,X}$ with $X>1$. Let, w.l.o.g., $u, v \in V_1$ lying on a horizontal line, $l_{h}$, at distance $1$ from each other. The $X$ vertices of $V_2$ of an empty-ply drawing are in $R_{uv_1} = R_{u_1}\cap R_{v_1}$  and 
$R_{uv_2} = R_{uv_1}^C$ (bounded by $2*dist(u, v)$ from both $u$ and $v$). 
(Figure \ref{figAppendix:RAB12}).

\begin{figure}
\centering
\includegraphics[width=0.5\textwidth]{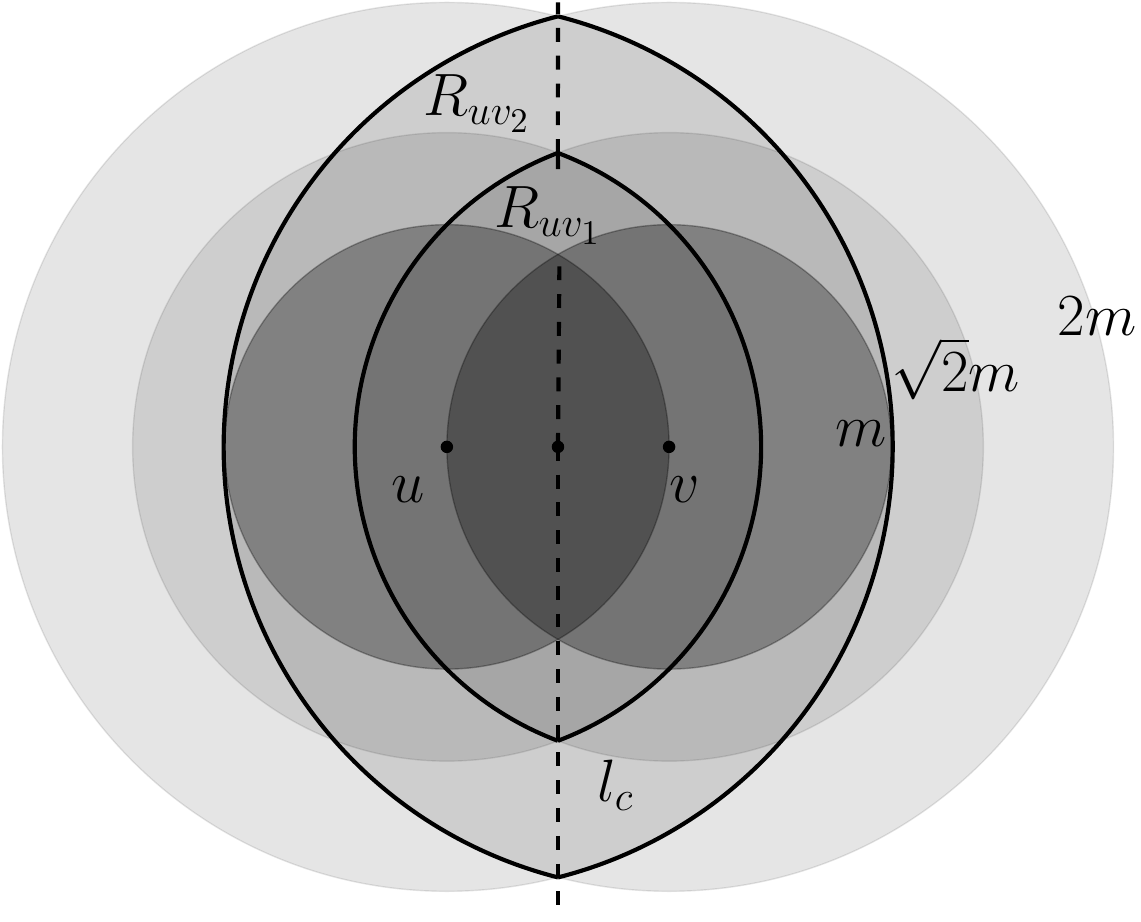}
\caption{The circles around $u$ and $v$ and the corresponding areas of $R_{uv_1}$ and $R_{uv_2}$.}\label{figAppendix:RAB12}
\end{figure}

Let $l_c$, the perpendicular bisector of $(u,v)$, split the plane in the left half-plane containing $u$ and the right half-plane containing $v$.
Each vertex $x_i \in V_2$, where $i = 1, \dots, |V_2|$,  in the left (resp. right) half-plane has $dist(u, x_i) < dist(v, x_i)$ (resp. $dist(u, x_i) > dist(v, x_i)$) while exactly on the line $dist(u, x_i) = dist(v, x_i)$. 
It follows that the ply disks of the vertices in the same half-plane depend only on either $u$ or $v$.

Also, we define for  $u$ and $v$ the cover disk of $u$ with respect to $v$, $Cov_{u}(v)$ (and viceversa), as the area that contains the vertices $x_i \in V_2$with  $i = 1, \dots, |V_2|$,  adjacent to $v$ whose ply disk has radius at least $\frac{dist(x_i, v)}{2}$ and therefore its disk covers $v$, i.e., drawing any vertex $x_i$ adjacent to $v$ in this area violates the empty-ply condition of the ratio of the incident edges.\\

We are now ready to prove the upper bound on the maximum number of vertices in $R_{uv_1}$ and $R_{uv_2}$.

\begin{lemma}\label{lem:vrab2}
There are at most 10 vertices in $R_{uv_2}$.
\end{lemma} 
\begin{proof}
We show that there can be at most $10$ sectors in $R_{uv_2}$. Let  w.l.o.g. $dist(u, v)=1$, for the property of an empty-ply drawing each vertex should be drawn in the intersection area of two disks with radius $2$. 
Since, as explained above, the ply disks of the vertices in the same half-plane depend on either $u$ or $v$, the number of vertices for an empty-ply drawing depends on the number of non-overlapping sectors starting from $u$ (resp. $v$) that can be placed in the right (resp. left) half-plane.
It is easy to see that the maximum angle to place the sectors is when the ply disks of $u$ and $v$ is $1$, i.e. there is at least one vertex at distance $2$ from $u$ and $v$.
In this case, $m=1$ and $R_{uv_2}=[\sqrt{2}, 2]$ and the available angle to place the sectors for each vertex is given by the points of $l_c$ at distance $2$ from $u$ and $v$ that is $151.04^\circ$ as shown by Figure~\ref{figAppendix:Ruv2beta}. It follows that the maximum number of vertices in $R_{uv_2}$ is given by $\lfloor\frac{\lceil 2 \beta_2\rceil}{\alpha} \rfloor = 10$.
\end{proof}

\begin{figure}
	\centering
	\includegraphics[width=.35\textwidth]{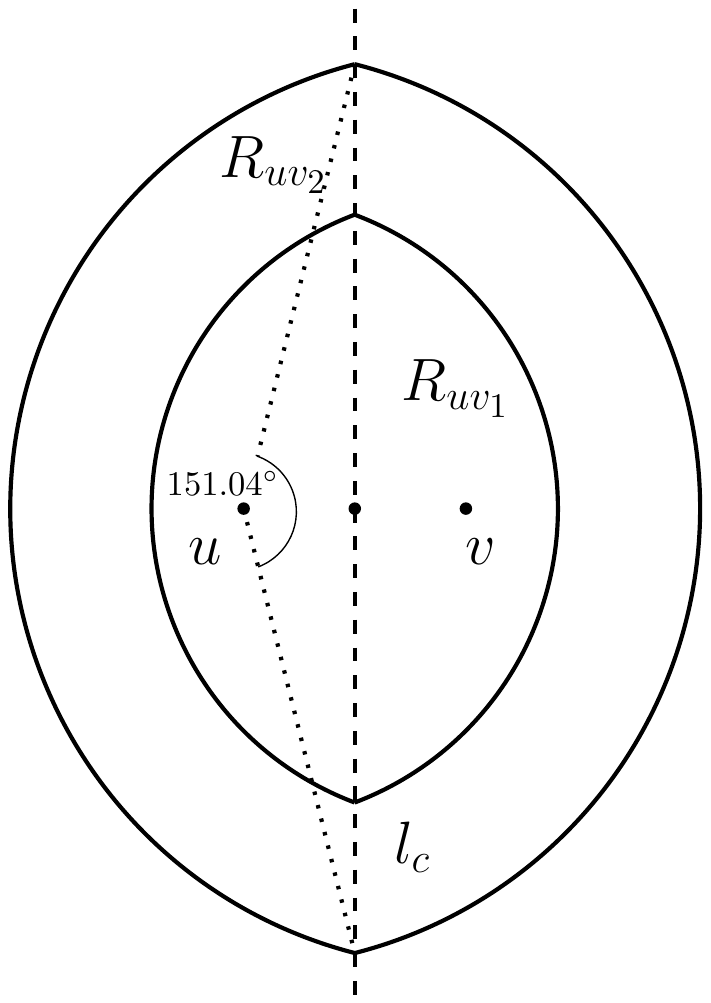}
	\caption{The maximum angle to place the sectors in $R_{uv_2}$ is $\beta_2 = 151.04^\circ$ that is when the radii of the ply disks of both $u$ and $v$ is the maximum, i.e. 1.}
	\label{figAppendix:Ruv2beta}
\end{figure}

Note that Lemma \ref{lem:sector} implies that there is a sector with $\gamma = (2 \beta_2) - (10 \alpha) = 302,08 - 278.9 =  23,14^\circ$ angle that is free from vertices. 

Although the previous Lemma gives an upper bound of $10$ sectors in $R_{uv_2}$ and, each of them can contain at most $1$ vertex per range, the actual number of sectors that can be placed in  $R_{uv_1}$ is different since the angle in this range, namely $\beta_1$, is smaller. 

Thus, we state the following Lemma.
\begin{lemma}\label{lem:vrab1}
In the area $R_{uv_1}$ there are at most $6$ vertices. 
\end{lemma}

\begin{proof}
Similar to the previous proof, the available angle from $u$ and $v$ in $R_{uv_1}$ should be considered with respect to $\sqrt{2}$, to maximize the number of sectors that can be placed. Also, differently from $R_{uv_2}$,  in this range the available space for the sectors is limited by the presence of $Cov_{u}(v)$ and $Cov_{v}(u)$ (these two disks do not influence the available space for the sectors in $R_{ab_2}$ at least at distance $2$ from $u$ and $v$). Taking into account this limitation, the available angle per each vertex $u$ and $v$ is $\beta_1 \leq 41.41^\circ$.
The available space can be split in two non intersecting parts: above and below $l_{h}$.
The number of sectors, and so vertices, that can be placed in each of these parts is given by $\lceil{\frac{2 \beta_1}{\alpha}}\rceil = \lceil2.97\rceil = 3$ and  so there cannot be $4$ vertices on either the top or the bottom side. It follows that the number of vertices in $R_{ab_1}$ is limited to $6$. 
\end{proof}

Lemma~\ref{lem:vrab2} and~\ref{lem:vrab1} imply that $K_{2, 17}$ does not admit an empty-ply drawing.

However, in the following we show that drawing the vertices in $R_{uv_1}$ and $R_{uv_2}$  the upper bound for the number of vertices of $V_2$ for an empty-ply drawing is smaller than $17$.

\begin{lemma} \label{lem:10-4}
If $R_{uv_2}$ contains exactly $10$ vertices then $R_{uv_1}$ contains at most $4$ vertices.
\end{lemma}

\begin{proof}
If $10$ vertices are drawn in $R_{uv_2}$
then $m \in [\frac{1.88}{2}, 1]$.  
In the inner range, the vertices have to lie in the intersection of the disks with radius $\sqrt{2}m$ centered in $u$ and $v$ and outside the of radius $\frac{1.88}{2}$ with the same centers. 
Let $l_h$ divide $R_{uv_1}$ the part above and below $l_h$. The ply disks of two vertices drawn one in the upper part and the other in the lower part do not overlap, thus we can argue about the two parts independently. Since the ply disk of any vertex drawn in the upper, or lower, part covers entirely its bisector, that is a  segment on $l_c$, it follows that in each part there can be at most $2$ vertices (i.e., one the left and one on the right of the bisector).
It follows that there can be at most $4$ vertices in $R_{uv_1}$ ($2$ in each part), if there are exactly $10$ vertices in $R_{uv_2}$.
\end{proof}

\begin{lemma}\label{lem:9-5}
If $R_{uv_2}$ contains exactly $9$ vertices then $R_{uv_1}$ contains at most $5$ vertices.
\end{lemma}

\begin{proof}
There exist $9$ vertices in $R_{uv_2}$
then $m \in [\frac{1.78}{2}, 1]$. 
If we draw three vertices in the range $[\frac{1.78}{2}, \sqrt{2}]$ there exists one vertex whose minimal distance
to either $u$ or $v$ is $\geq 1.33$. Any vertex with distance $\geq 1.33$ implies a cover-disk regarding $u$ and $v$ which covers a sector of at least $45.44^\circ$. This angle can not be used to draw any vertex in the range $[\sqrt{2}m, 2m]$. Thus $3$ vertices in the upper, or lower, part of $[m, \sqrt{2}m]$ imply at most $9$ vertices.
It follows that there cannot be $9$ vertices in $R_{uv_2}$ and  $6$ in $R_{uv_1}$.
\end{proof}

Lemma~\ref{lem:10-4} and ~\ref{lem:9-5} imply Theorem~\ref{thm:bipartiteK2X}

\section*{Appendix B}

\rephrase{Theorem}{\ref{thm:ternary-shrinking}}{
	For no $q\in(0,1)$, rooted ternary trees admit empty-ply drawings constructed in orthogonal fashion with shrink factor $q$, i.e. when the distance of a vertex to its children is $q$ times the distance to its parent.
}

\begin{proof}
	Consider the sequence of centers $v_0,v_1,\dots$ as shown in Fig.~\ref{fig:limit} with its limit $w$.
	\begin{figure}[t]
		\scalebox{0.6}{\includegraphics{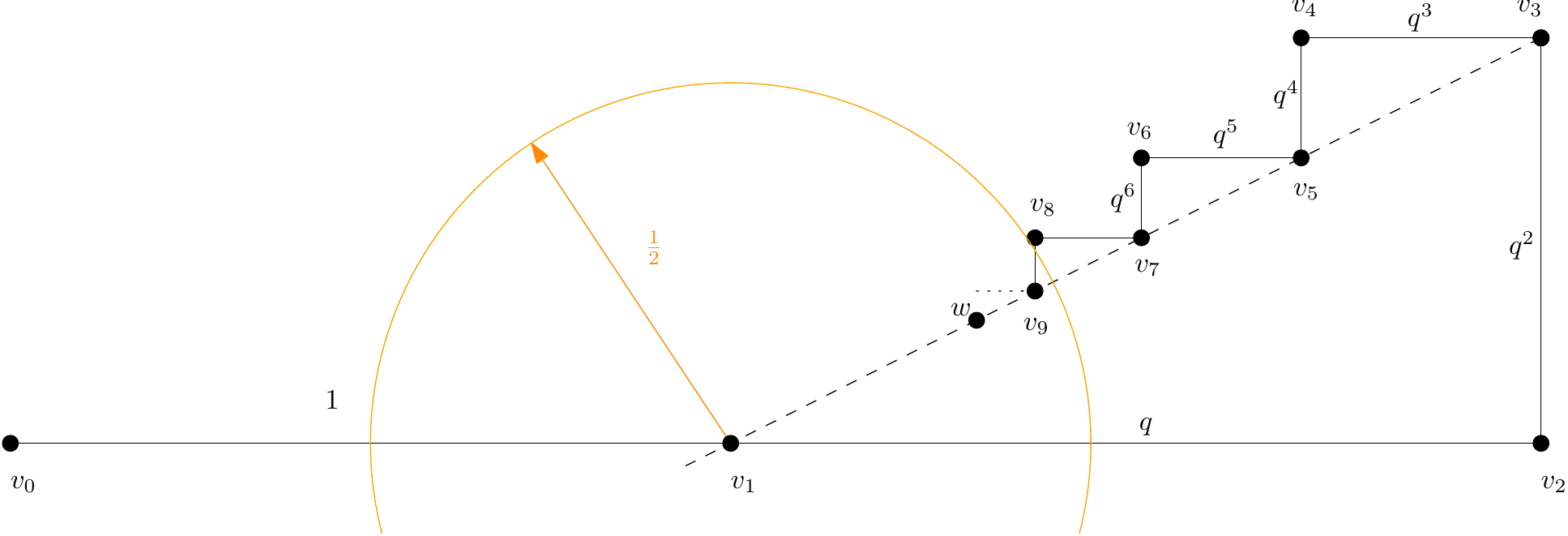}}
		\caption{The sequence of centers used in the proof of Theorem~\ref{thm:ternary-shrinking}.\label{fig:limit}}
	\end{figure}
	The Euclidean distance between $v_3$ and $w$ is:
	$$|v_3 w|=\frac{q^3}{1-q^2}\sqrt{1+q^2}~~~\Rightarrow~~~|v_1 w|=\left(q-\frac{q^3}{1-q^2}\right)\sqrt{1+q^2}$$
	If the following inequality is not satisfied for some scaling factor $q\in(0,1)$, then for any $\varepsilon > 0$, infinitely many points of the sequence $v_0,v_1,\dots$ will belong to the $\varepsilon$-neighborhood of~$w$, and consequently inside the disk centered at $v_1$.
	$$|v_1 w|=\left(q-\frac{q^3}{1-q^2}\right)\sqrt{1+q^2}\ge \frac12$$
	Since the function $f(q)=\left(q-\frac{q^3}{1-q^2}\right)\sqrt{1+q^2}-\frac12$ is negative in $(0,1)$, we get 
	$|v_1 w|<\frac{1}{2}$. It follows that $w$ lies inside the ply disk of $v_1$. Consequently, infinitely many points of the sequence $v_0,v_1,\dots$ lie inside the ply disk of $v_1$.
\end{proof}

\end{document}